\documentclass[envcountsame,envcountsect,orivec,runningheads]{llncs}

\usepackage[hypertexnames=false]{hyperref}
\usepackage{amsmath,amsfonts,amssymb,amscd}
\usepackage{latexsym}
\usepackage{breakurl} 
\newfont{\bbb}{bbm10 scaled 1100}         % blackboard bold
\newcommand{\IN}{\mbox{\bbb N}}          % set of natural numbers
\usepackage{enumerate}
\usepackage{graphicx}
\usepackage{xspace}
\DeclareSymbolFont{frenchscript}{OMS}{ztmcm}{m}{n}
\usepackage{wrapfig}

\usepackage{xcolor}

\spnewtheorem{observation}[theorem]{Observation}{\bfseries}{\itshape}

\newcommand{\Thm}[1]{Theorem~\ref{thm:#1}}
\newcommand{\Cor}[1]{Corollary~\ref{cor:#1}}
\newcommand{\Prop}[1]{Proposition~\ref{prop:#1}}
\newcommand{\Lem}[1]{Lemma~\ref{lem:#1}}
\newcommand{\Def}[1]{Definition~\ref{df:#1}}
\newcommand{\Ex}[1]{Example~\ref{ex:#1}}
\newcommand{\Obs}[1]{Observation~\ref{obs:#1}}
\newcommand{\App}[1]{Appendix~\ref{sec:#1}}
\newcommand{\Sect}[1]{Section~\ref{sec:#1}}
\newcommand{\SSect}[1]{Section~\ref{ssec:#1}}

\newcommand{\Tab}[1]{Table~\ref{tab:#1}}
\newcommand{\Eq}[1]{\eqref{eq:#1}}
\newcommand{\mylabel}[1]{\hypertarget{lab:#1}{\ \mbox{{\scriptsize\sc (#1)}}}}
\newcommand{\myref}[1]{\hyperlink{lab:#1}{\sc (#1)}\xspace}

\DeclareMathSymbol{\A}{\mathord}{frenchscript}{65}   % set of ABC agent identifiers
\DeclareMathSymbol{\B}{\mathord}{frenchscript}{66}   % set of ABC broadcast names
\DeclareMathSymbol{\C}{\mathord}{frenchscript}{67}   % set of ABC handshake communication names
\DeclareMathSymbol{\FS}{\mathord}{frenchscript}{70}  % fairness specification
\DeclareMathSymbol{\HC}{\mathord}{frenchscript}{72}  % set of ABC handshake communications
\DeclareMathSymbol{\Lab}{\mathord}{frenchscript}{76} % set of labels in transition system
\newcommand{\T}{{\rm Ex}}                         % set of closed terms
\newcommand{\E}{P}                                   % typical process expression
\newcommand{\F}{Q}                                   % typical process expression
\newcommand{\AI}{A}                                  % typical agent identifier
\newcommand{\BI}{B}                                  % typical agent identifier
\newcommand{\CI}{C}                                  % typical agent identifier

\makeatletter
\def\moverlay{\mathpalette\mov@rlay}
\def\mov@rlay#1#2{\leavevmode\vtop{%
   \baselineskip\z@skip \lineskiplimit-\maxdimen
   \ialign{\hfil$\m@th#1##$\hfil\cr#2\crcr}}}
\newcommand{\charfusion}[3][\mathord]{
    #1{\ifx#1\mathop\vphantom{#2}\fi
        \mathpalette\mov@rlay{#2\cr#3}
      }
    \ifx#1\mathop\expandafter\displaylimits\fi}
\makeatother
\newcommand{\dcup}{\charfusion[\mathbin]{\cup}{\mbox{\Large$\cdot$}}}
\newcommand{\plat}[1]{\raisebox{0pt}[0pt][0pt]{#1}}  % no vertical space
\newcommand{\bis}{\raisebox{.3ex}{$\underline{\makebox[.7em]{$\leftrightarrow$}}$}}
\newcommand{\ABC}{ABC\xspace}
\newcommand{\en}{\textit{en}}
\newcommand{\src}{{\it src}}
\newcommand{\target}{{\it target}}
\newcommand{\shar}[1]{\mathord{\stackrel{#1}{\rightarrow}}}
\makeatletter
\def\comesfrom{\@transition\leftarrowfill}
\def\goesto{\@transition\rightarrowfill}
\def\ngoesto{\@transition\nrightarrowfill}
\def\Goesto{\@transition\Rightarrowfill}
\def\nGoesto{\@transition\nRightarrowfill}
\def\xmapsto{\@transition\mapstofill}
\def\nxmapsto{\@transition\nmapstofill}
\def\@transition#1{\@@transition{#1}}
\newbox\@transbox
\newbox\@arrowbox
\newbox\@downbox
\def\@@transition#1#2%
   {\setbox\@transbox\hbox
      {\vrule height 1.5ex depth .9ex width 0ex\hskip0.25em$\scriptstyle#2$\hskip0.25em}
   \ifdim\wd\@transbox<1.5em
      \setbox\@transbox\hbox to 1.5em{\hfil\box\@transbox\hfil}\fi
   \setbox\@arrowbox\hbox to \wd\@transbox{#1}
   \ht\@arrowbox\z@\dp\@arrowbox\z@
   \setbox\@transbox\hbox{$\mathop{\box\@arrowbox}\limits^{\box\@transbox}$}
   \dp\@transbox\z@\ht\@transbox 10pt
   \mathrel{\box\@transbox}}
\def\nrightarrowfill{$\m@th\mathord-\mkern-6mu%
  \cleaders\hbox{$\mkern-2mu\mathord-\mkern-2mu$}\hfill
  \mkern-6mu\mathord\not\mkern-2mu\mathord\rightarrow$}
\def\Rightarrowfill{$\m@th\mathord=\mkern-6mu%
  \cleaders\hbox{$\mkern-2mu\mathord=\mkern-2mu$}\hfill
  \mkern-6mu\mathord\Rightarrow$}
\def\nRightarrowfill{$\m@th\mathord=\mkern-6mu%
  \cleaders\hbox{$\mkern-2mu\mathord=\mkern-2mu$}\hfill
  \mkern-6mu\mathord\not\mathord\Rightarrow$}
\def\mapstofill{$\m@th\mathord\mapstochar\mathord-\mkern-6mu%
  \cleaders\hbox{$\mkern-2mu\mathord-\mkern-2mu$}\hfill
  \mkern-6mu\mathord\rightarrow$}
\def\nmapstofill{$\m@th\mathord\mapstochar\mathord-\mkern-6mu%
  \cleaders\hbox{$\mkern-2mu\mathord-\mkern-2mu$}\hfill
  \mkern-6mu\mathord\not\mkern-2mu\mathord\rightarrow$}
\makeatother %%%%% end of arrows definition
\newcommand{\ar}[1]{\mathrel{\goesto{#1}}}            % arrow
\newcommand{\nar}[1]{\mathrel{\ngoesto{#1\;}}}        % negated arrow

\makeatother

\newcounter{Hequation}

\makeatletter
\g@addto@macro\equation{\stepcounter{Hequation}}
\makeatother
\begin{document}

\title{Progress, Fairness and Justness
\texorpdfstring{\\}{}in Process Algebra}
\titlerunning{Progress, Fairness and Justness}
\authorrunning{R.J.\ van Glabbeek \& P.\ H\"ofner}
\author{
Rob van Glabbeek \inst{1}\inst{2}\and
Peter H\"ofner\inst{1}\inst{2}
}
\institute{
NICTA\thanks{NICTA is funded by the Australian Government through the
    Department of Communications and the Australian Research Council
    through the ICT Centre of Excellence Program.}, Sydney, Australia\and
Computer Science and Engineering, UNSW, Sydney, Australia
}
\maketitle

\setcounter{footnote}{0}

\begin{abstract}
To prove liveness properties of concurrent systems, it is often necessary to postulate progress, fairness and justness properties.
This paper investigates how the necessary progress, fairness and justness assumptions can be added to or incorporated in a standard
process-algebraic specification formalism. We propose a formalisation that can be applied to a wide range of
process algebras.  The presented formalism is used to reason about route discovery and packet delivery in the setting of wireless networks.
 \end{abstract}

%%%%%%%%%%%%%%%%%%%%%%%%%%%%%%%%%%%%%%%%%%
\section{Introduction}\label{sec:introduction}
%%%%%%%%%%%%%%%%%%%%%%%%%%%%%%%%%%%%%%%%%%
In a process-algebraic setting, safety properties of concurrent systems
are usually shown by the use of invariants on a labelled transition system (LTS).
This does not require any assumptions about the behaviour of
concurrent systems beyond their modelling as states in an LTS\@.
In order to prove liveness properties on the other hand
it is usually necessary to postulate certain progress, fairness and justness
properties as part of the specification of the systems under investigation.
This paper investigates how the necessary progress, fairness and justness
properties can be added to or incorporated in a standard process-algebraic specification formalism. 
Liveness properties are formalised in terms of a temporal logic interpreted on
complete paths in the LTS of the process algebra.
Progress, fairness and justness properties are captured by fine-tuning the definition of what
constitutes a complete path.

\Sect{ABC} introduces an Algebra of Broadcast Communication 
(\ABC)---a variant of the process algebra CBS \cite{CBS91}---that 
is essentially CCS \cite{Mi89} augmented with a
formalism for broadcast communication. \ABC is given a structural operational
semantics \cite{Pl04} that interprets expressions as states in an LTS\@.
We develop our approach for formalising liveness properties as
well as progress, fairness and justness assumptions in terms of this
process algebra. 
However, the presented approach can be applied to a wide range of process algebras.
\ABC is largely designed to be a convenient starting point for 
transferring the presented theory to such algebras; it
contains all the features for which we are aware that the application of our theory poses non-trivial problems,
and, at the same time, is kept as simple as possible.
In \cite{TR13} we apply the same approach to a more involved
process algebra called AWN (Algebra for Wireless Networks). 

\Sect{properties} recalls Linear-time Temporal Logic (LTL)~\cite{Pnueli77}
and describes a way to interpret it on a labelled transition system that
arises as the semantic interpretation of a process algebra like \ABC.
This yields a way to represent desirable properties of concurrent
systems specified in such a process algebra by means of LTL properties.
We illustrate this by formulating \emph{packet delivery},
a liveness property studied in \cite{TR13} in the context of wireless mesh network protocols.
The presented development applies just as well to desirable properties of concurrent
systems specified in branching time temporal logics such as CTL and CTL${}^{*}$.

In \SSect{progress} we formulate an elementary \emph{progress assumption} on the
behaviour of processes, without which no useful liveness property of a system
will hold. In the standard interpretation of temporal logic
\cite{Pnueli77,Emerson91} a stronger progress assumption is built in, but we
argue that this stronger version is not a valid assumption in the context of
reactive systems. In order to derive a progress assumption that is both
necessary and justifiable in the reactive context we envision, we introduce the
concept of an \emph{output action}, which cannot be blocked by the context in
which a process is running. Although output actions are commonplace in
many specification formalisms, their use in process algebra is limited at best,
and we have not seen them used to define progress properties.
The main reason for working with a language that is
richer than CCS, is that restricted to CCS the set of output actions would be empty.

In \SSect{fairness} we discuss \emph{weak} and \emph{strong} \emph{fairness assumptions} 
and propose a formalisation in the context of process algebras like \ABC by
augmenting a process-algebraic specification $P$ with a \emph{fairness specification},
which is given as a collection of temporal logic formulas. This follows the
traditional approach of TLA~\cite{TLA} and other formalisms \cite{Fr86}, ``in which
first the legal computations are specified, and then a fairness notion
is used to exclude some computations which otherwise would be legal'' \cite{AFK88}.
However, in order to do justice to the reactive nature of the systems under consideration,
we need a more involved consistency requirement between the
process-algebraic specification of a system and its fairness specification.

In \SSect{justness} we propose a \emph{justness assumption} for parallel-composed
transition systems, essentially assuming progress of all the component processes.
In the literature, such justness properties are typically seen as special cases of
weak fairness properties, and the term \emph{justice} is often used as a 
synonym for \emph{weak fairness}. Here we consider justness to be a notion
distinct from fairness, and propose a completely different formalisation.
Fairness is a property of schedulers that repeatedly choose between tasks, 
  whereas justness is a property of parallel-composed transition systems.
Nevertheless, we show that our notion of justness coincides with the original notion of justice of
\cite{LPS81}---a weak fairness property. This requires an interpretation of the work of
\cite{LPS81} applied to LTSs involving a more precise definition---and decision---%
of what it means for a transition to be continuously enabled.

Finally, \Sect{fairness implementations} addresses the question whether system
specifications consisting of a process-algebraic and a fairness specification
allow implementations that can be described entirely process-algebraic, i.e.,
without fairness component. Here an `implementation' allows the replacement of
nondeterministic choices by (more) deterministic choices following a particular
scheduling policy.
In the context of CCS we
conjecture a negative answer by showing an extremely simple
fair scheduling specification---given as a CCS expression augmented with a
fairness specification---that could not be implemented by any CCS expression alone.
This specification does allow an implementation in \ABC without fairness
 component, which takes advantage of justness properties for output actions.

%%%%%%%%%%%%%%%%%%%%%%%%%%%%%%%%%%%%%%%%%%
\section{\ABC---An Algebra of Broadcast Communication}\label{sec:ABC}
%%%%%%%%%%%%%%%%%%%%%%%%%%%%%%%%%%%%%%%%%%
\begin{table*}[t]
\normalsize
\centering
\caption{Structural operational semantics of \ABC}
\label{tab:ABC}
\framebox{
$\begin{array}{@{}ccc@{}}
\alpha.\E \ar{\alpha} \E  \mylabel{Act} &
\displaystyle\frac{\E\ar{\alpha} \E'}{\E+\F \ar{\alpha} \E'}  \mylabel{Sum-l}&
\displaystyle\frac{\F \ar{\alpha} \F'}{\E+\F \ar{\alpha} \F'}  \mylabel{Sum-r}
\\[3ex]
\displaystyle\frac{\E\ar{\eta} \E'}{\E|\F \ar{\eta} \E'|\F} \mylabel{Par-l}&
\displaystyle\frac{\E\ar{c} \E' ,~ \F \ar{\bar{c}} \F'}{\E|\F \ar{\tau} \E'| \F'} \mylabel{Comm}&
\displaystyle\frac{\F \ar{\eta} \F'}{\E|\F \ar{\eta} \E|\F'} \mylabel{Par-r}
\\[3ex]
\!\displaystyle\frac{\E\ar{b\sharp_1} \E' ,~ \F \nar{b?}}{\E|\F \ar{b\sharp_1} \E'| \F}  \mylabel{Bro-l}&
\displaystyle\frac{\E\ar{b\sharp_1} \E' ,~ \F \ar{b\sharp_2} \F'}{\E|\F \ar{b\sharp} \E'| \F'} \mylabel{Bro-c} &
\displaystyle\frac{\E\nar{b?} ,~ \F \ar{b\sharp_2} \F'}{\E|\F \ar{b\sharp_2} \E| \F'}  \mylabel{Bro-r}\\[-5pt]
& \scriptstyle \sharp_1\circ\sharp_2 = \sharp \neq \_ ~~\mbox{with}~~
    \begin{array}{c@{\ }|@{\ }c@{\ \ }c}
    \scriptstyle \circ & \scriptstyle ! & \scriptstyle ? \\
    \hline
    \scriptstyle ! & \scriptstyle \_ & \scriptstyle ! \\
    \scriptstyle ?& \scriptstyle ! & \scriptstyle ? \\
    \end{array}
\\[1ex]
\displaystyle\frac{\E \ar{\ell} \E'}{\E[f] \ar{f(\ell)} \E'[f]} \mylabel{Rel}&
\displaystyle\frac{\E \ar{\ell} \E'}{\E\backslash c \ar{\ell} \E'\backslash c}~(c\mathord{\neq}\ell\mathord{\neq}\bar{c}) \mylabel{Res}&
\displaystyle\frac{\E \ar{\ell} \E'}{A\ar{\ell}\E'}~(A \mathord{\stackrel{{\it def}}{=}} P) \mylabel{Rec}
\end{array}$}
\end{table*}

The Algebra of Broadcast Communication (\ABC) is parametrised with sets ${\A}$ of \emph{agent identifiers},
$\B$ of \emph{broadcast names} and $\C$ of \emph{handshake communication names};
each $\AI\in\A$ comes with a defining equation \plat{$\AI \stackrel{{\it def}}{=} P$}
with $P$ being a guarded \ABC expression as defined below.

The collections $\B!$ and $\B?$ of \emph{broadcast} and \emph{receive}
actions are given by $\B\sharp:=\{b\sharp \mid b\mathbin\in\B\}$ for $\sharp \in \{!,?\}$.
The set $\bar{\C}$ of \emph{handshake communication co-names} is $\bar\C:=\{\bar{c} \mid c\in {\C}\}$,
and
the set $\HC$ of \emph{handshake actions} is 
\plat{$\HC:=\C \dcup \bar\C$}, the disjoint union of the names and co-names.
The function
$\bar{\rule{0pt}{7pt}.}$ is extended to $\HC$ by
declaring $\bar{\bar{\mbox{$c$}}}=c$.

Finally, \plat{$Act := \B! \dcup \B? \dcup \HC\dcup \{\tau\}$} is the set of \emph{actions}.
Below, $\AI$, $\!\BI$, $\!\CI$ range over $\A\!$, $b$ over $\B$, $c$ over $\HC\!$,
$\eta$ over $\HC\cup\{\tau\}$ and $\alpha,\ell$ over $Act$.
A \emph{relabelling} is a function $f\!:(\B\mathbin\rightarrow \B) \cup (\C\mathbin\rightarrow \C)$.
It extends to $Act$ by $f(\bar{c})\mathbin=\overline{f(c)}$, $f(b\sharp)\mathbin=f(b)\sharp$ and $f(\tau):=\tau$.
The set $\T_{\rm \ABC}$ of \ABC expressions is the smallest set including:

\vspace{-1.2mm}
\begin{center}
\begin{tabular}{@{}l@{~~}l@{\qquad\quad}l@{~~}l@{\qquad\quad}l@{~~}l@{}}
$0$ & \emph{inaction}&
$\alpha.\E$  & \emph{prefixing}&
$\E+\F$  & \emph{choice} \\
$\E|\F$ & \emph{parallel composition}&
$\E\backslash c$  & \emph{restriction} &
$\E[f]$ &  \emph{relabelling} \\
$\AI$ &  \emph{agent identifier}\\
\end{tabular}
\end{center}
\vspace{-1.2mm}

\noindent for $\E,\F\in\T_{\rm \ABC}$ and relabellings $f$.
An expression is guarded if each agent identifier occurs within the scope of a prefixing operator.
The semantics of \ABC is given by the labelled transition relation
$\mathord\rightarrow_{\rm \ABC} \subseteq \T_{\rm \ABC}\times Act \times\T_{\rm \ABC}$, where the transitions 
\plat{$\E\ar{\ell}\F$}
are derived from the rules of \Tab{ABC}.

\ABC is basically the Calculus of Communicating Processes (CCS)~\cite{Mi89} augmented with a
formalism for broadcast communication taken from the Calculus of Broadcasting Systems (CBS)~\cite{CBS91}.
The syntax without the broadcast and receive actions and all rules except \myref{Bro-l},
\myref{Bro-c} and \myref{Bro-r} are taken verbatim from CCS\@. However, 
the rules now cover the different name spaces; \myref{Act} for example allows labels of
broadcast and receive actions. The rule \myref{Bro-c}---without rules
like \myref{Par-l} and \myref{Par-r} with label $b!$---implements a
form of broadcast communication where any broadcast $b!$
performed by a component in a parallel composition is guaranteed to be received by any other
component that is ready to do so, i.e., in a state that admits a $b?$-transition.
In order to ensure associativity of the parallel composition, one also needs this 
rule for components receiving at the same time ($\sharp_1\mathord=\sharp_2\mathord=\mathord{?}$).
The rules \myref{Bro-l} and \myref{Bro-r} are added to make
broadcast communication \emph{non-blocking}: without them a component could be delayed in
performing a broadcast simply because one of the other components is not ready to receive it.

\begin{theorem}
Strong bisimilarity \cite{Mi89} is a congruence for all operators of \ABC.
\end{theorem}

\begin{proof}
This follows immediately from the observation that all rules of \Tab{ABC} are in the GSOS format of
Bloom, Istrail \& Meyer, using \cite[Theorem 5.1.2]{BIM95}.
\qed\end{proof}
To establish the associativity of parallel composition of \ABC up to strong bisimilarity ($\bis{}$), we will employ a general result of
Cranen, Mousavi \& Reniers~\cite{CMR08}. However, for this result to apply, we need a structural
operational semantics of the language in the De Simone format \cite{dS85}---so without negative premises.

To this end, let $\B\mathord{:} :=\{b\mathord: \mid b\in\B\}$ be the set of \emph{broadcast discards}, and
\plat{$\Lab := \B\mathord: \dcup Act$} the set of \emph{transition labels}. We enrich the
transition relation of \ABC with transitions labelled with discard communications, by adding the rules\\[1mm]
\centerline{$\begin{array}{@{}ccc@{}}
0 \ar{b:} 0 \mylabel{Dis0}&\qquad
\alpha.\E \ar{b:} \alpha.\E ~~ \mbox{($\alpha \mathord{\neq} b?$)} \mylabel{Dis1} &\qquad
\displaystyle\frac{\E\ar{b:} \E' ,~ \F \ar{b:} \F'}{\E+\F \ar{b:} \E'+ \F'} \mylabel{Dis2}
\end{array}$}\\
to \Tab{ABC}, allowing $\sharp_1\mathord=\sharp_2\mathord=\sharp\mathord=\mathord:$ in
\myref{Bro-c},%
\footnote{The remaining cases are still undefined, i.e,\ 
$\sharp_1\circ : \ =\  :\circ \sharp_2 = \_$ (for $\sharp_1, \sharp_2\not=\ :$).
} and letting $\ell$ range over all of $\Lab$.

\begin{lemma} {\rm \cite{CBS91}}
$P \!\ar{b:}\! Q$ iff $Q\mathop=P \wedge P\! \nar{b?}$\,, 
for $P,Q\mathop\in\T_{\rm \ABC}$ and $b\mathop\in\B$.
\end{lemma}

\begin{proof}
A straightforward induction on derivability of transitions.
\qed\end{proof}

\noindent\begin{minipage}{0.84\textwidth}
Because of this, a negative premise $P \nar{b?}$ can be replaced by a positive premise $P \ar{b:} P'$
and all rules {\sc (Bro)} in \Tab{ABC} can be unified into the single rule \myref{Bro-c},
where $\sharp_1,\sharp_2,\sharp$ range over $\{!,?,:\}$
and $\circ$ is defined by the table on the right.
The resulting rules are all in the De Simone format.
\end{minipage}
\begin{minipage}{0.15\textwidth}
\vspace{-1mm}
$~~\begin{array}{c@{\ }|@{\ }c@{\ \ }c@{\ \ }c}
     \circ &  ! &  ? &  : \\
    \hline
     ! &  \_ &  ! &  ! \\
     ?&  ! &  ? &  ? \\
     : &  ! &  ? &  :
    \end{array}$
\end{minipage}

\begin{corollary}\label{cor:modified same}
$\!$The original and modified structural operational semantics of \ABC yield the same labelled transition
relation $\rightarrow_{\ABC}$ when transitions
labelled $b\mathord{:}$ are ignored.
\end{corollary}
In fact, our `modified' operational semantics stems directly from CBS \cite{CBS91}.

\begin{theorem}
In \ABC, parallel composition is associative up to $\bis{}$.
\end{theorem}

\begin{proof}
The associativity depends on the generated transition relation only, and is preserved when
ignoring transitions with a particular label.
So by \Cor{modified same} it suffices to investigate the modified semantics.
The modified operational rules of \ABC  fit the \emph{\sf ASSOC-De Simone} rule format of
\cite{CMR08}, which guarantees associativity up to $\bis{}$.
The detailed proof that our rules fit this format is similar to the proof of Theorem 4.4 in \cite{TR13}.
\qed\end{proof}

%%%%%%%%%%%%%%%%%%%%%%%%%%%%%%%%%%%%%%%%%%
\section{Formalising Temporal Properties}\label{sec:properties}
%%%%%%%%%%%%%%%%%%%%%%%%%%%%%%%%%%%%%%%%%%

We will use Linear-time Temporal Logic (LTL)~\cite{Pnueli77} to specify properties that one would
like to establish for concurrent systems. For the purpose of this paper, any other temporal logic could have been used as well.

We briefly recapitulate the syntax and semantics of LTL; a thorough and formal 
introduction to this logic can be found e.g.\ in~\cite{HuthRyan04}. The logic is built from a set of
\emph{atomic propositions} that characterise facts that may hold in some state of a (concurrent) system.
A classical example is that a `transition $\nu$ is enabled', denoted by $\en(\nu)$.

LTL formulas are interpreted on paths in a transition system,\footnote{A 
\emph{transition system} is given by a set $S$ of \emph{states} and a set $T\subseteq S\times S$ of 
\emph{transitions}.} where each state is labelled with the atomic propositions that hold in that state.
A \emph{path} $\pi$ is a finite or infinite sequence of states such there is a transition from each
state in $\pi$ to the next, except the last one if $\pi$ is finite.
An atomic proposition $p$ holds for a path $\pi$ if $p$ holds in the first state of $\pi$.

LTL \cite{Pnueli77} uses the temporal operators $\mathbf{X}$, $\mathbf{G}$, $\mathbf{F}$ and $\mathbf{U}$.%
\footnote{$\mathbf{X}$ and $\mathbf{U}$ were not introduced in the original paper~\cite{Pnueli77};
  they were added later on.}
The formulas $\mathbf{X} \phi$, $\mathbf{G} \phi$ and $\mathbf{F} \phi$
mean that $\phi$ holds in the second state on a given path, \emph{globally} in all states,
and \emph{eventually} in some state, respectively;
$\phi\mathbf{U}\psi$ means that $\psi$ will hold eventually, and
$\phi$ holds in all states \emph{until} this happens.%
\footnote{$\mathbf{G}$ and $\mathbf{F}$ can be expressed in terms of $\mathbf{U}$:
  $\mathbf{F} \phi \equiv {\tt true}\mathbf{U}\phi$ and $\mathbf{G} \phi \equiv
  \neg\mathbf{F} \neg\phi$.} 
Here a formula $\phi$ is deemed to
\emph{hold in a state on a path $\pi$} iff it holds for the remainder of $\pi$ when starting from that state.
LTL formulas can be combined by the logical connectives \emph{conjunction} $\wedge$,
\emph{disjunction} $\vee$, \emph{implication} $\Rightarrow$ and \mbox{\emph{negation} $\neg$}.

An LTL formula holds for a process, modelled as a state in a transition system iff it holds
for all \emph{complete} paths in the system starting from that state.\footnote{A path staring
    from a state $s$ is also called a path \emph{of} $s$.}
A path is complete iff it leaves no transitions undone without a good reason;
in the original work on LTL \cite{MP92}
 the complete paths are exactly the infinite ones, but in
\Sect{progress} we propose a different concept of completeness:
a path will be considered complete iff it is \emph{progressing}, \emph{fair} and \emph{just},
as defined in Sections \ref{ssec:progress}, \ref{ssec:fairness} and \ref{ssec:justness}, respectively.%
\footnote{We declare a formula $\mathbf{X} \phi$ false on any path that lacks a second state.}

Below we will apply LTL to the labelled transition system\/ $\mathcal{T}$ generated by the structural
operational semantics of \ABC.\footnote{A \emph{labelled transition system} (LTS) is given by a set $S$ of
  states and a transition relation $T\subseteq S\times \Lab \times S$ for some set of labels $\Lab$.
  The LTS generated by \ABC has $S:=\T_{\rm \ABC}$ and $T:=\mathord{\rightarrow_{\rm \ABC}}$.
}
Here, the most natural atomic propositions are the transition labels: they tell when an action takes place.
These propositions hold for transitions rather than for states.
Additionally, one can consider state-based propositions such as $\en(\nu)$.
In languages that maintain data variables,\footnote{such as the algebra for wireless networks AWN \cite{TR13}; see below}
propositions such as `$x\leq7$' that report on the current value of such variables can also be associated to the states.

To incorporate the transition-based atomic propositions into the framework of temporal logic, we
perform a translation of the transition-labelled transition system\/ $\mathcal{T}$ into a
state-labelled transition system\/ $\mathcal{S}$, and apply LTL to the latter.  A suitable
translation, proposed in \cite{DV95}, introduces new states halfway the existing
transitions---thereby splitting a transition $\ell$ into $\ell;\tau$---and attaches transition
labels to the new `midway' states.
If we also have state-based atomic propositions stemming from\/ $\mathcal{T}$, we furthermore
declare any atomic proposition except $\en(\nu)$ that holds in state $Q$ to also hold for the new state midway a
transition $P\ar{\ell}Q$.
LTL formulas are interpreted on the paths in\/ $\mathcal{S}$. Such a path is a sequence of states of\/ $\mathcal{S}$,
and thus an alternating sequence of states and transitions from\/ $\mathcal{T}$. Here we will only
consider paths that are infinite or end in a state of\/ $\mathcal{T}$; paths ending `midway a
transition' will not be complete, progressing, fair or just.

Below we use LTL to formalise properties that say that whenever a precondition $\phi^{\it pre}$ holds
in a reachable state, the system will eventually reach a state satisfying the postcondition $\phi^{\it post}$.
Such a property is called an \emph{eventuality property} in \cite{Pnueli77}; it is formalised by the LTL formula
$\mathbf{G} \big(\phi^{\it pre} \Rightarrow \mathbf{F}\phi^{\it post}\big).$

\begin{example}
In a language like \ABC we can model a network by means of a parallel composition of
processes running on the nodes in the network. Each of those processes $A_i$, for $1\mathbin\leq i\mathbin\leq n$,
could be specified by a defining equation \plat{$A_i \stackrel{{\it def}}{=} P_i$}, where $P_i$ always ends
with a recursive call to $A_i$. This way, the behaviour specified by $P_i$ is repeated forever.
The processes $A_i$ send messages to each other along shared channels. Here a message $m$
transmitted along a broadcast or handshake channel is modelled by a name $c_m \mathbin\in \B$ or $c_m \mathbin\in \C$.

Suppose the process $A_0$ can receive messages $m \in \{1,...,k\}$ from the environment.
This could be modelled by \plat{$A_0 \stackrel{{\it def}}{=} c_1.P_0^1+\cdots + c_k.P_0^k$}.
The behaviour of the nodes in the network could be specified so as to guarantee that such a message
will eventually reach the node running the process $A_n$, which will deliver it to the environment by
performing the broadcast $d_m!$. We may assume that no other nodes can perform the actions $c_m$
or $d_m!$.

A useful property that this network should have is \emph{packet delivery}:
any message received from the environment by $A_0$ will eventually be delivered back to the
environment by $A_n$.
In LTL it can be formulated as $\textbf{G}(c_m \Rightarrow \textbf{F} d_m!)$.
\end{example}
In \cite{TR13} we model a routing protocol in a process algebra for wireless networks (AWN) that
captures dynamic topologies, where nodes drift in and out of transmission range, and communication
between two nodes is successful only when they are within transmission range of each other.
In this context a packet delivery property is formulated that can be obtained from a property like the
one above by incorporating a number of side conditions.

%%%%%%%%%%%%%%%%%%%%%%%%%%%%%%%%%%%%
\section{Progress, Fairness and Justness}\label{sec:progress}
%%%%%%%%%%%%%%%%%%%%%%%%%%%%%%%%%%%%
%in case we want to rename the labels
\newcommand{\Figprogress}{Figure~\ref{fig:PJF}(a)\xspace}
\newcommand{\Figjustness}{Figure~\ref{fig:PJF}(c)\xspace}
\newcommand{\Figfairness}{Figure~\ref{fig:PJF}(b)\xspace}

In \cite[Sect.~9]{TR13}, as well as above, we formalise properties that say that under certain conditions some desired activity
will eventually happen, or some desired state will eventually be reached. 
As a simple instance consider the transition systems in Figures~\ref{fig:PJF}(a)--(c), where
the double-circled state satisfies a desired property $\phi$.  The formula $\mathbf{G} (a \Rightarrow
\mathbf{F}\phi)$ says that once the action $a$ occurs, eventually we will reach a state where $\phi$
holds.  We investigate reasons why this formula might not hold, and formulate
assumptions that guarantee it does.

\subsection{Progress}\label{ssec:progress}
The first thing that can go wrong is that the process of \Figprogress\footnote{Following the approach
  of CCS \cite{Mi89} we identify processes and states, and do not use a notion of an initial state.
  When speaking of a process depicted graphically, by default we mean the state indicted by the short
  arrow in the figure.}
performs $a$, thereby reaching the state $s$,
and subsequently remains in the state $s$ without ever performing the internal action $\tau$ that leads to
the desired state $t$, satisfying $\phi$. If there is the possibility of remaining in a state
even when there are enabled internal actions, no useful liveness property about
processes will ever be guaranteed.
We therefore make an assumption that rules out this type of~behaviour.%
\begin{equation}\tag{$P_{1}$}\label{eq:progress1}
    \parbox{0.85\textwidth}{\textit{A process in a state that admits an internal
    action\footnotemark\ will eventually perform an action.}}
\end{equation}
\footnotetext{\ABC offers only one internal action $\tau$.
  Any other action is called \emph{external}.}%
\eqref{eq:progress1} is called a \emph{progress} property. It guarantees that the process depicted in
\Figprogress satisfies the LTL formula $\mathbf{G} (a \Rightarrow \mathbf{F}\phi)$.
We cannot assume progress when only external actions are possible. For instance, the process of
\Figprogress will not necessarily perform the action $a$,
and hence need not satisfy the formula $\mathbf{F}\phi$. The reason is that external actions
could be synchronisations with the environment, and the environment may not be ready to synchronise.
In \ABC this can happen if $a$ is a handshake action $c\in\HC$
or a receive action $b?\in\B?$.
Here it makes sense to distinguish two kinds of external actions: those whose
execution requires cooperation from the environment in which the process runs, and those that do not.
We call the latter kind \emph{output actions}. As far as progress properties go, output
actions can be treated just like internal actions:
\begin{equation}\tag{$P_{2}$}\label{eq:progress2}
    \parbox{0.85\textwidth}{\textit{A process in a state that admits an output action will
    eventually perform an action.}}
\end{equation}
In case $a$ is an output action, which can happen independent of the environment, the
formula $\mathbf{F}\phi$
 holds for the process of \Figprogress.
In the remainder we treat internal actions and output actions together; we call them \emph{non-blocking} actions.

We formalise \eqref{eq:progress1} and \eqref{eq:progress2} through a suitable modification of the definition of a
complete path.  In early work on
temporal logic, formulas were interpreted on Kripke structures: transition systems with unlabelled transitions,
subject to the condition of \emph{totality}, saying that each state admits at least one outgoing
transition. In this context, the complete paths are defined to be all infinite paths
of the transition system. When giving up totality, it is customary to deem complete
also those paths that end in a state from which no further transitions are possible~\cite{DV95}. 
Here, we go a step further, and consider paths that are either infinite
or end in a state from which no further non-blocking actions are possible. 
Those paths are called \emph{progressing}. This
definition exactly captures \eqref{eq:progress1} and~\eqref{eq:progress2}.

This proposal is a middle ground between two extremes.
Dropping all progress properties amounts to defining \emph{each} path to be complete.
This yields a temporal logic that is not powerful enough to establish
nontrivial eventuality properties.
Defining a path to be complete only when it cannot be further extended, on the other hand,
incorporates progress properties
that do not hold for reactive systems.
It would give rise to the unwarranted conclusion that the property $\mathbf{F}\phi$
holds for the process of \Figprogress, regardless of the nature of $a$.

\begin{figure}[t]
	\centering
		\begin{minipage}[b]{0.43\linewidth}
			\centering
				\includegraphics[scale=0.88]{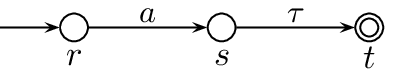}\\[-1mm]
				{\small (a)\ Progress}
		\end{minipage}
		\hspace{0.03\linewidth}
		\begin{minipage}[b]{0.43\linewidth}
			\centering
				\includegraphics[scale=0.88]{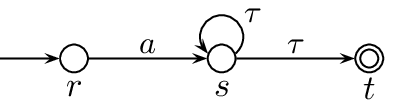}\\[-1mm]
				{\small (b)\ Fairness}
		\end{minipage}\\[1.3ex]
		\begin{minipage}[b]{0.97\linewidth}
			\centering
				\includegraphics[scale=0.88]{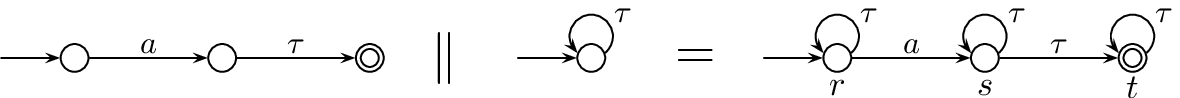}\\[-1mm]
				{\small (c)\ Justness}
		\end{minipage}
		\vspace{-1mm}
	\caption{Progress, Fairness and Justness}
	\label{fig:PJF}
	\vspace{-3mm}
  \end{figure}

As we will show, progressing paths do not capture fairness and justness properties;
hence they should not be called complete. In Sections~\ref{ssec:fairness} and
\ref{ssec:justness} we will propose a notion of a complete path
that is progressing and also captures such properties.
This restriction concerns the infinite paths only; for finite paths
complete will coincide with progressing.

It remains to decide which of the external actions generated by the structural operational semantics
of a language should be classified as output actions. 
Some actions cannot be output, since they can be blocked by the environment. 
In \ABC, any handshake action $c\in \HC$ can 
be blocked by restriction.
Since $	(c.0)\backslash c$ cannot perform any action, $c$ cannot
be output.
For the remaining external actions, the user can decide whether they are output or not.
For \ABC we decide that a broadcast ($b!$) is an output action, whereas a receive ($b?$)
is not.
Informal intuition explains the reason:  a process that broadcasts a message should be able to perform this action 
independent of
whether other processes are receiving it. 
On the other hand, a process should only be able to receive a message that \emph{is} sent by another process.

The above analysis assumes the original operational semantics of \ABC, with negative premises.
For the modified semantics with discard-transitions, the question arises whether
$b\mathord:$ should count as an output action. Here the right answer is that $b\mathord:$ is a
transition label that does not count as an action at all. The reason is that we want our progress
property \Eq{progress2} to imply that the \ABC process $b_1!.b_2!.0$ will eventually execute the
broadcast action $b_2$. However, a potentially complete path that invalidates this property consists
of $b_1!$ followed by infinitely many broadcast discards $b_1\mathord:$ (or $b_2\mathord:$).
In the original operational semantics such a path does not exist, and the property is satisfied.
To obtain the same result in the modified operational semantics, we classify $b\mathord:$ as a
non-action. This way, \Eq{progress2} says that after
$b_1\mathord:$ the process will eventually perform a transition that is not a discard; this must be $b_2\mathord{!}$.
Formally, this is achieved by excluding from the definition of progressing path all paths that end in
infinitely many discard-transitions (all looping in the final state), where the final state in the path
admits a non-blocking action.
To avoid further encounters with this complication, we will
 henceforth assume the original operational semantics.

\subsection{Fairness}\label{ssec:fairness}
With the progress requirements \Eq{progress1} and \Eq{progress2}
embedded in our semantics of LTL, the process
of \Figprogress satisfies the formula $\mathbf{G} (a \Rightarrow \mathbf{F}\phi)$.
Yet, the process of \Figfairness does not satisfy this formula.
The reason is that in state $s$ a choice is made between two internal transitions.
One leads to the desired state satisfying $\phi$, whereas the other gives the process a chance to
make the decision again. This can go wrong in exactly one way, namely if the
$\tau$-loop is chosen every time.

For some applications it is warranted to make a \emph{global fairness assumption}, saying that in
verifications we may simply assume our processes to eventually escape from a loop such as in
\Figfairness and do the right thing. A process-algebraic verification approach based on such an
assumption is described in~\cite{BBK87a}. Moreover, a global fairness assumption is
incorporated in the weak bisimulation semantics employed in~\cite{Mi89}.
Different global fairness assumptions in process algebra appear in \cite{CS87}.

An alternative approach, which we follow here, is to explicitly declare certain choices to be fair,
while leaving open the possibility that others are not.
A \emph{strong fairness}
assumption requires that if a task is enabled infinitely often,\footnote{or in the final
  state of a run, although for many tasks this is a logical impossibility} but
allowing interruptions during which it is not enabled, it will eventually be scheduled.
Such a property is expressed in LTL as
$\mathbf{G}(\mathbf{GF}\psi \Rightarrow \mathbf{F}\phi)$,%
\footnote{These properties were introduced in LTL in \cite{GPSS80} under the  name `responsiveness to persistence'.}
or equivalently $\mathbf{G}\mathbf{F}\psi \Rightarrow \mathbf{G}\mathbf{F}\phi$; here $\psi$ is the condition that states that
the task is enabled, whereas $\phi$ states that it is being
executed.
\hypertarget{weakfairness}{A \emph{weak fairness} assumption requires that if a
task, from some point onwards, is perpetually enabled, it will eventually be scheduled. In LTL
this is expressed as $\mathbf{G}(\mathbf{G}\psi \Rightarrow \mathbf{F}\phi)$,}\footnote{These properties were
  introduced in LTL in \cite{GPSS80} under the  name `responsiveness to insistence',
  and deemed `the minimal fairness requirement' for any scheduler.}
or equivalently  $\mathbf{F}\mathbf{G}\psi \Rightarrow \mathbf{G}\mathbf{F}\phi$.

If a formula $\psi$ holds in state $s$ of \Figfairness, and $\phi$ in $t$,
then the strong fairness as\-sumption $\mathbf{G}(\mathbf{GF}\psi \Rightarrow \mathbf{F}\phi)$ ensures
the choice at $s$ to be fair. If $\psi$ even holds in (or during) the transition that constitutes the
$\tau$-loop, the weak fairness assumption $\mathbf{G}(\mathbf{G}\psi \Rightarrow \mathbf{F}\phi)$ suffices.
If this property is part of the specification, the process of
 \Figfairness will satisfy the desired eventually property $\mathbf{G} (a \Rightarrow \mathbf{F}\phi)$.

In general, we propose a specification framework where a process is specified by a pair of a
process expression $P$ (for instance in the language \ABC) and a \emph{fairness specification} $\FS$, consisting of a
collection of LTL formulas (where for instance the actions of \ABC are allowed as atomic propositions).
Typically, $\FS$ contains strong or weak fairness properties.

The semantics of such a specification is again a pair.
The first component is the state $P$ in the
LTS generated by \ABC, and the second component, the set of \emph{fair paths}, is a subset of the
progressing paths starting from $P$, namely those that satisfy the formulas in $\FS$.

We require the state $P$ in the LTS and the fairness specification to be consistent with each other.
By this we mean that from $P$ one cannot reach a state from where, given a sufficiently uncooperative
environment, it is impossible to satisfy the fairness specification---in other words \cite{Lamport00},
`the automaton can never ``paint itself into a corner''\;\!'. In~\cite{Lamport00,TLA} this requirement is
called \emph{machine closure}, and demands that any finite path in the LTS, starting from $P$, can
be extended to a path satisfying $\FS\!$.
Since we deal with a reactive system here, we need a more involved consistency requirement,
taking into account all possibilities of the environment to allow or block transitions that
are not fully controlled by the specified system itself. This requirement can best be explained in
terms of a two player game between a \emph{scheduler} and the \emph{environment}.

The game begins with any finite path $\pi$ starting from $P$, ending in a state $Q\in\T_{\rm \ABC}$, chosen by the environment.
In each turn, first the environment selects a set $\textit{next}(Q)$ of transitions originating from
$Q$; this set has to include all transitions labelled with non-blocking actions originating
from $Q$, but can also include further transitions starting from $Q$. If $\textit{next}(Q)$ is
empty, the game ends; otherwise the scheduler selects a transition from this set, which is, together
with its
target state, appended to $\pi$, and a new turn starts with the prolonged finite path.
The \emph{result} of the game is the finite path in which the game ends, or---if it
does not---the infinite path that arises as the limit of all finite paths encountered
during the game. The game is \emph{won} by the scheduler iff the result is a
progressing\footnote{When adopting our proposal of \SSect{justness}, the resulting path should even be just.}
path that satisfies $\FS$.
Now $P$ is \emph{consistent} with $\FS$ iff there exists a winning strategy for the scheduler.

\subsection{Justness}\label{ssec:justness}
Now suppose we have two concurrent systems that work independently in parallel, such as two
completely disconnected nodes in a network. One of them is modelled by the transition system of
\Figprogress, and the other is doing internal transitions in perpetuity.  The parallel
composition is depicted on the left-hand side of \Figjustness. According to our structural operational
semantics, the overall transition system resulting from this parallel composition is the one
depicted on the right. In this transition system, the LTL formula $\mathbf{G} (a \Rightarrow
\mathbf{F}\phi)$ is not valid, because, after performing the action $a$, the process may do an
infinite sequence of internal transitions that stem from the `right' component in the parallel
composition, instead of the transition to the desired success state.  Yet the formula $\mathbf{G} (a
\Rightarrow \mathbf{F}\phi)$ does hold intuitively, because no amount of internal activity in the
right node should prevent the left component from making progress.  That this formula does not
hold can be seen as a pitfall stemming from the use of interleaving semantics.  The intended
behaviour of the process is captured by the following \emph{justness} property:%
\begin{equation}\tag{$J$}\label{eq:justness}
\parbox{0.85\textwidth}{\textit{%
If a combination of components in a parallel composition is in a state
    that admits a non-blocking action, then one (or more) of them  will eventually partake in an action.
}}
\end{equation}

Thus justness guarantees progress of all components in a parallel composition, and of all
combinations of such components. In the \ABC expression $((P|Q)\backslash a) | R$ for instance,
we might reach a state where $P$ admits an action $c\mathbin\in\HC$ with $c\mathbin{\neq} a$ and $R$ admits
$\bar c$. Thereby, the combination of these components admits an action $\tau$. Our justness
assumption now requires that the combination of $P$ and $R$ will eventually perform an action.
This could be the $\tau$-action obtained from synchronising $c$ and $\bar c$, but it also could be
any other action from either $P$ or $R$. 

Note that progress is a special case of justness, obtained by considering any process as the
combination of all its parallel components.

\newcommand{\startingfrom}{of\xspace}
We now formalise the justness requirement \eqref{eq:justness}.
\\
Any transition \plat{$P|Q \ar{\ell} R$} derives, through the
rules of \Tab{ABC}, from
\begin{itemize}
\vspace{-1ex}
\item a transition \plat{$P \ar{\ell} P'$} and a state $Q$, where $R=P'|Q$\,,
\item two transitions \plat{$P \ar{\ell_1} P'$ and $Q \ar{\ell_2} Q'$}, where $R=P'|Q'$\,,
\item or from a state $P$ and a transition \plat{$Q \ar{\ell} Q'$}, where $R=P|Q'$.
\vspace{-1ex}
\end{itemize}
This transition/state, transition/transition or state/transition pair is called a 
\emph{decomposition}
of \plat{$P|Q \ar{\ell} R$}; it need not be unique.
Now a \emph{decomposition} of a path $\pi$ \startingfrom $P|Q$ into paths $\pi_1$ and $\pi_2$
\startingfrom $P$ and $Q$, respectively, is obtained by decomposing each transition in the path, and
concatenating all left-projections into a path \startingfrom $P$ and all right-projections into a
path \startingfrom $Q$---notation $\pi \in \pi_1 | \pi_2$.
Here it could be that $\pi$ is infinite, yet either $\pi_1$ or $\pi_2$ (but not both) are finite.
Again, decomposition of paths need not be unique.

Likewise, any transition \plat{$P[f] \ar{\ell} R$} stems from a transition \plat{$P \ar{\ell_1} P'$}, where $R=P'[f]$.
This transition is called a decomposition of $P[f] \ar{\ell} R$. A \emph{decomposition} of a path
$\pi$ \startingfrom $P[f]$  is obtained by decomposing each transition in the path, and
concatenating all transitions so obtained into a path \startingfrom $P$.
In the same way one defines a decomposition of a path \startingfrom $P\backslash c$.

We now define a path of a process to be \emph{just} if it models a run that can actually occur in
some environment, even when postulating \Eq{justness}; we call it $Y\!$-\emph{just}, for $Y\subseteq\HC$, if it can
occur in an environment that from some point onwards blocks all actions in $Y\cup\B?$.
\begin{definition}\label{df:just path}\rm
\emph{$Y\!$-justness}, for $Y\mathbin\subseteq\HC$, is the largest family of predicates on the paths in the
transition system\/ $\mathcal{S}$ associated to the LTS\/ $\mathcal{T}$ of \ABC such that
\begin{itemize}
\vspace{-1ex}
\item a finite $\hspace{-.5pt}Y\!$-just path ends in a state of $\mathcal{T}\!$ that admits actions from $\hspace{-.5pt}Y\mathord\cup\B?$ only;
\item a $Y\!$-just path of a process $P|Q$ can be decomposed into an $X$-just path of $P$ and a $Z$-just
  path of $Q$ such that $Y\mathbin\supseteq X\mathord\cup Z$ and $X\mathord\cap \bar{Z}\mathbin=\emptyset$---here
  $\bar Z\mathbin{:=}\{\bar{c} \hspace{-.5pt}\mid\hspace{-.5pt} c\mathbin\in Z\}$;
\item a $Y\!$-just path of
  $P\backslash c$ can be decomposed into a $Y\mathord\cup\{c,\bar c\}$-just path of $P$;
\item a $Y\!$-just path of $P[f]$ can be decomposed into an
  $f^{-1}(Y)$-just path of $P$;
\item and each suffix of a $Y\!$-just path is $Y\!$-just.
\vspace{-1ex}
\end{itemize}
A path is \emph{just} if it is $Y\!$-just for some $Y\subseteq\HC$.
\end{definition}
The last clause in the second requirement prevents an $X$-just path of $P$ and a $Z$-just path of
$Q$ to compose into an $X\mathord\cup Z$-just path of $P|Q$ when $X$ contains an action $c$ and $Z$
the complementary action~$\bar c$. The reason is that no environment can block both actions for
their respective components, as nothing can prevent them from synchronising with each other.
The fifth requirement helps characterising processes of the form $b.0+(A|b.0)$ and $a.(A|b.0)$, with  \plat{$A \stackrel{{\it def}}{=} a.A$}.
Here, the first transition `gets rid of' the choice and of the leading action $a$, respectively, 
and reduces the justness of paths of such processes to their suffixes.

If $Y\subseteq Z$ then any $Y\!$-just path is also $Z\!$-just. As a consequence, 
a path is just iff it is  $\HC\!$-just.
In \App{just progressing} we show
that a finite path is just iff it does not end in a state
from which a non-blocking action is possible, i.e., iff it is progressing as defined in \SSect{progress}.

A path is called \emph{complete} if it is fair as well as just, and hence also progressing.

The above definition of a just path captures our (progress and) justness requirement, and
ensures that the formula $\mathbf{G} (a \Rightarrow \mathbf{F}\phi)$ holds for the process of \Figjustness.
For example, the infinite path $\pi$ starting from $r$ that after the $a$-transition keeps looping through the
$\tau$-loop at $s$ can only be derived as $\pi_1 | \pi_2$, where $\pi_1$ is a finite path ending
right after the $a$-transition. Since $\pi_1$ fails to be just (its end state
admits a $\tau$-transition), $\pi$ fails to be just too, and hence does not count when searching for
a complete path that fails to satisfy $\mathbf{G} (a \Rightarrow \mathbf{F}\phi)$.

\subsection{Justness versus Justice}\label{ssec:justice}

The concept of \emph{justice} was introduced in \cite{LPS81}: `A computation is said to be
\emph{just} if it is finite or if every transition which is continuously enabled beyond a certain
point is taken infinitely many times.\!' In LTL this amounts to
$\mathbf{F}\mathbf{G}\,\en(\nu) \mathbin\Rightarrow
\mathbf{G}\mathbf{F}\,\nu$ for each transition $\nu$,
thus casting justice as a weak fairness property.
\advance\textheight 3pt

In \cite{LPS81} the identity of a transition, when appearing in a parallel composition,
is not affected by the current state of the parallel component.
For instance, the two transitions $c.0|d.0 \ar{c} 0|d.0$ and $c.0|0 \ar{c} 0|0$\ ---they differ in
their source and target states---are seen as the same transition of the process $c.0|d.0$, stemming
from the left component and scheduled either before or after the $d$-transition of the right component.
In \App{enabling}, to be read after \ref{sec:derivations} and~\ref{sec:concurrency},
we introduce the notion of an \emph{abstract transition}---an
equivalence class of \emph{concrete transitions}---to formalise the transitions intended in \cite{LPS81}.

\advance\textheight -3pt
In the context of reactive systems, an (abstract) transition $\nu$ typically is a synchronisation between the
system and its environment. In case the environment does not synchronise,
$\nu$ cannot happen, even when it is continuously enabled. For this reason, here
justice is only reasonable for abstract
transitions $\nu$ labelled with non-blocking actions.

In applying the concept of justice from \cite{LPS81} to LTSs, there is potential ambiguity in
what counts as `continuously'. Consider the \ABC system specified by \plat{$B \stackrel{{\it def}}{=} c.B + b!.0$}.
By \Def{just path}, the computation consisting of $c$s only is just; it satisfies \Eq{justness}. 
However, it could be argued that $b!$ is continuously enabled. 
This would make the computation unjust in the sense of  \cite{LPS81}.
On the other hand, the choice between $c$ and $b!$ may be
non-deterministic, and could always be resolved in favour of~$c$.
Therefore we do not consider
this computation unjust, and adopt the principle of `noninstantaneous readiness'
\cite{AFK88}, stating that the enabledness of the
$b!$ is interrupted when performing the $c$-transition. In our model, this is implemented by the midway
states corresponding with transitions. As a result, we judge the specified execution just, and hence
do not claim that $b!$ will happen eventually.

\newcommand{\nub}{\ensuremath{\nu_{b!}}}
On the other hand, in our vision the enabledness of a transition cannot be interrupted by performing a
concurrent transition.
For instance, the execution $c^\omega$ of the process $C|b!.0$, where \plat{$C\stackrel{{\it def}}{=}c.C$},
is unjust, because the $b!$-transition $\nub$ is continuously enabled and never taken.
In Appendices~\ref{sec:concurrency} and~\ref{sec:enabling} we formalise this by a
novel definition of the predicates $\en(\nu)$,
such that $\en(\nub)$ holds during the transition $C|b!.0 \ar{c} C|b!.0$.

In doing so, we have to overcome a problem illustrated by the process $C|B$ with $B$ and $C$ as above.
Whether the path \plat{$C|B \ar{c} C|B \ar{c} \dots$} counts as being just by the mantra of \cite{LPS81}
depends on whether $\en(\nub)$ holds during each transition
\plat{$C|B \ar{c} C|B$} in that path. This, in turn, depends on whether these transitions
originate from $C$, so that they are concurrent with $\nub$, or from $B$.
We formalise this by using a richer transition system\/  $\mathcal{U}$ in which
the two transitions $C|B \ar{c} C|B$ are distinguished.
The states of\/ $\mathcal{U}$ are the states of\/ $\mathcal{T}$ together with the \emph{derivations} of
transitions of\/ $\mathcal{T}$ from the rules of \Tab{ABC}---the latter are the \emph{concrete transitions}
alluded to above. The transitions of\/ $\mathcal{U}$ are
$P \rightarrow \chi$ and $\chi \rightarrow Q$, for any
derivation $\chi$ of a transition $P\ar{\alpha}Q$.
The predicates $\en(\nu)$ are defined on the states of\/ $\mathcal{U}$.
The transition system\/ $\mathcal{S}$ associated to the LTS\/ $\mathcal{T}$ of \ABC
can be obtained from\/ $\mathcal{U}$ by consistently identifying multiple derivations of the same transition.
Now, any path $\pi$ in\/ $\mathcal{U}$ projects onto a path
$\widehat\pi$ in\/ $\mathcal{S}$,
and any path in\/ $\mathcal{S}$ is of the form $\widehat\pi$.
Details can be found in \App{derivations}.

In the literature \cite{Emerson91,MP92}, the concept of
\emph{weak fairness} often occurs as a synonym for ``justice''.
At the same time, the potential ambiguity in what counts as
being continuously enabled is resolved differently from the approach
we take here: a transition that from some point onwards is enabled \emph{in every
state} cannot be ignored forever. Under this notion of weak fairness, the
system $B$ discussed above will surely perform the $b!$-action.
It would be useful to have different names for a concept of justice or weak
fairness that adopts the principle of noninstantaneous readiness and
one that does not.

The following theorem states that
the former concept of justice is in perfect agreement with our
  notion of justness of \SSect{justness}.
Its proof can be found in
 \App{proof}.

\begin{theorem}\label{thm:just path}
A path of an \ABC process is just in the sense of \Def{just path} iff it is of the form
$\widehat\pi$ for a path $\pi$ in\/ $\mathcal{U}$ that satisfies
the LTL formulas $\mathbf{F}\mathbf{G}\,\en(\nu) \Rightarrow \mathbf{G}\mathbf{F}\, \nu$ for each
abstract transition $\nu$ with $\ell(\nu)\in Act$  a non-blocking  action.
\end{theorem}

%%%%%%%%%%%%%%%%%%%%%%%%%%%%%%%%%%%%%%%%%%
\section{Implementing Fairness Specifications}\label{sec:fairness implementations} 
%%%%%%%%%%%%%%%%%%%%%%%%%%%%%%%%%%%%%%%%%%
%redefinition of process names 
\newcommand{\II}{I}
\renewcommand{\AI}{\II_{1}}                            % typical agent identifier
\renewcommand{\BI}{G}                                  % typical agent identifier
\renewcommand{\CI}{\II_{2}}                            % typical agent identifier
\renewcommand{\d}[1]{c_{#1}}
\newcommand{\db}[1]{c_{{#1}}!}
\newcommand{\dr}[1]{c_{#1}?}
\newcommand{\tb}[1]{t_{#1}}
\newcommand{\tr}[1]{t_{#1}?}
\newcommand{\e}{e}
\newcommand{\BB}[1]{\BI_{#1}}

For certain properties of the form $(\bigvee_{i}\mathbf{G}\mathbf{F}a_{i}) \Rightarrow (\bigvee_{j}\mathbf{G}\mathbf{F}b_{j})$
where the $a_i$ and $b_i$ are action occurrences%
---hence for specific strong fairness properties---one can
define a \emph{fairness operator} that
transforms a given LTS into a LTS that satisfies the property~\cite{PuhakkaValmari01}.
This is done by eliminating all the paths that do not satisfy the property via a carefully designed parallel composition. 
In the same vein, we ask whether any process specification involving a fairness specification can
be implemented by means of a process-algebraic expression without fairness component. Here
we give an example that we believe cannot be implemented in
standard process algebras like CCS\@. To make this more precise, let CCS$^!$ be the fragment of \ABC
without receive actions; equivalently, this is the fragment of CCS in which certain names $b$ induce
no co-names $\bar b$ and no restriction operators $\backslash b$. These actions are deemed output
actions, meaning that we do not consider environments that can prevent them from occurring.

Consider the CCS$^!$ process
$
(\AI\,|\,\BI\,|\,\CI)\backslash \d1\backslash\d2
$, 
where\\[0mm]
\centerline{$
\II_{i}\stackrel{\it def}{=} r_i.\bar{c_i}.\II_{i}\,\quad (i\in\{1,2\})
\ \mbox{ and}\qquad
\BI\stackrel{\it def}{=}c_1.\tb1.\e.\BI  \,+\, c_2.\tb2.\e.\BI 
$}\\[1mm]
augmented with the fairness specification
$
\bigwedge_{i=1,2}\mathbf{G}(r_i \Rightarrow \mathbf{F}(t_i!))
$.

Here $t_1$, $t_2$, $e$ are output actions.
This process could be called a \emph{fair scheduler}. The actions $r_1$ and $r_2$ can be seen as
requests received from the environment to perform tasks $\tb1$ and $\tb2$, respectively.
Each $r_i$ triggers a 
task $t_i$.
Moreover, between each two occurrences of
$t_i$ and $t_j$
an action $\e$ needs to be scheduled.

\begin{trivlist} \item[\hspace{\labelsep}{\bf Conjecture 1.}]\it
There does not exist a CCS\/${}^!$ expression $\BI$ such that the process
$
(\AI\,|\,\BI\,|\,\CI)\backslash \d1\backslash\d2
$,
with {$\AI$ and $\CI$} as above, has the following properties:
\begin{enumerate}
\vspace{-1ex}
\item On each complete (= just) path, each $r_i$ is followed by a $t_i$.
\item On each finite path no more $t_i$s than $r_i$s occur.
\item Between each pair of occurrences of $t_i$ and $t_j$ ($i,j\mathbin\in\{1,2\}$) an action $\e$ occurs.%
\vspace{-1ex}
\end{enumerate}
\end{trivlist}
We use CCS$^!$ rather than CCS to prevent the environment invalidating \textit{1.}\ by disallowing $t_i$.
We believe that there is no way to encode a fair scheduler with these properties in CCS$^!$ without the help of a
fairness specification.

However, we can do it in \ABC:\\[1mm]
\centerline{$
\begin{array}{l@{}cl@{\qquad\qquad}l@{}cl}
\AI&\stackrel{\it def}{=}&r_1.\db1.\AI &
\CI&\stackrel{\it def}{=}&r_2.\db2.\CI\\[1ex]
\BI &\stackrel{\it def}{=}&\dr1.\BB1 + \dr2.\BB2&
\BI' &\stackrel{\it def}{=}&\e.\BI + \dr1.\BB1' + \dr2.\BB2' \\
\BB{i} &\stackrel{\it def}{=}&\dr{j}.\BB{ij} + \tb{i}.\BI'&
\BB{i}' &\stackrel{\it def}{=}&\e.\BB{i} + \dr{j}.\BB{ij}'\\
\BB{ij} &\stackrel{\it def}{=}&\tb{i}.\BB{j}'&
\BB{ij}' &\stackrel{\it def}{=}&\e.\BB{ij}\\
\end{array}
$}\\[1mm]
with $i,j\in\{1,2\}$ and $i\not= j$.
This scheduler satisfies the fairness specification since the justness properties for output
actions require that once $r_i$ occurs, $c_{i}!$ must follow, and then
$t_i$ will eventually happen, at the latest when $\BI_{ij}$ is reached.

Currently, it is an open question whether arbitrary fairness specifications can be implemented in \ABC.

%%%%%%%%%%%%%%%%%%%%%%%%%%%%%%%%%%%%%%%%%%
\section{Conclusion and Outlook}\label{sec:conclusion}
%%%%%%%%%%%%%%%%%%%%%%%%%%%%%%%%%%%%%%%%%%
In this paper we have investigated how progress, fairness and justness
assumptions
can be handled within a process-algebraic specification formalism. 
Our semantics of a process is a state $P$ in an LTS together with a set of complete paths:
paths \startingfrom $P$ that are \emph{progressing}, \emph{fair} and \emph{just}.
We specify the fair paths by means of temporal logic, using a \emph{fairness specification} in
addition to a process-algebraic expression $P$. The progressing and just paths, on the other hand, are completely
determined by the syntax of $P$.

To demonstrate that the introduced approach is not only a theoretical result, we have applied 
the formalism  to a more involved process algebra called AWN (Algebra for Wireless Networks) and 
analysed the IETF-standardised \emph{Ad hoc On-demand Distance Vector (AODV) routing protocol}~\cite{RFC3561}.
We investigated two fundamental properties of routing protocols: \emph{route discovery} and \emph{packet delivery}.
Route discovery---a property that every routing protocol ought to satisfy---states that if a route discovery process is initiated in a state where the source is connected to the destination and no (relevant) link breaks, then the source will eventually discover a route to the destination.
Surprisingly, using the presented mechanism we could show that this property does not hold. The second property, packet delivery, was already sketched in \Sect{properties}; 
it has been shown that this property does not hold either. As a consequence, AODV does not satisfy two of the most crucial properties of routing protocols.
Details can be found in \cite{TR13}.
The  formalisation of progress, fairness and justness presented here was crucial for these results; 
without making these assumptions, no routing protocol would satisfy the route discovery and
packet delivery properties.

Future work will include the definition of suitable semantic equivalences on an LTS 
together with a set of complete paths, and their algebraic characterisations.

\newpage
\def\SSort#1{}\def\NSort#1{}

\newpage
\begin{appendix}
\setcounter{secnumdepth}{3}

\section*{\Large Appendices}
\section{Finite Paths are Just iff they are Progressing}\label{sec:just progressing}

\begin{proposition}\label{prop:progressing}
A finite path in\/ $\mathcal{S}$ is $Y\!$-just, for $Y\mathbin\subseteq \HC$, iff its last state is a state $Q\in\T_{\rm\ABC}$
of\/ $\mathcal{T}$ and all transitions enabled in $Q$ are labelled with actions from $Y\cup\B?$.
\end{proposition}
\begin{proof}
\hspace{-2pt}``$\Rightarrow$'': This follows immediately from the first requirement of \Def{just path}.

\newcommand{\fjust}{just$_{\text{fin}}$}
``$\Leftarrow$'': Define a path in\/ $\mathcal{S}$ to be $Y\!$-\emph{\fjust} if it is finite, its last state
is a state $Q\in\T_{\rm\ABC}$ of\/ $\mathcal{T}$, and all transitions enabled in $Q$ are labelled with actions from $Y\cup\B?$.
Then the family of predicates $Y\!$-justness${}_{\text{fin}}$, for $Y\mathbin\subseteq \HC$,
satisfies the five requirements of \Def{just path}. Since $Y\!$-justness is the largest family of
predicates satisfying those requirements, $Y\!$-justness${}_{\text{fin}}$ implies $Y\!$-justness.
\qed\end{proof}
It follows that a finite path is just iff it is progressing.

\section{A Concrete Kripke Structure for ABC}\label{sec:derivations}

In \Sect{properties} we extracted a transition system\/ $\mathcal{S}$ with unlabelled transitions---a
\emph{Kripke structure} \cite{HuthRyan04}, but without the condition of totality---out of the LTS\/
$\mathcal{T}$ generated by the structural operational semantics of ABC\@.
The states of\/ $\mathcal{S}$ are the states of\/ $\mathcal{T}$, that is, $\T_{\rm\ABC}$, together with
the transitions $P \ar{\alpha} Q$ of\/ $\mathcal{T}$. The transitions of\/ $\mathcal{S}$ are
$P \rightarrow (P \ar{\alpha} Q)$ and $(P \ar{\alpha} Q)\rightarrow Q$, for any transition $P \ar{\alpha} Q$ of\/ $\mathcal{T}$.
Next, we would like to define predicates $\en(\nu)$ on the states of\/
$\mathcal{S}$ indicating whether an (abstract) transition $\nu$ is enabled in a state $s$ of\/ $\mathcal{S}$.
If $s$ is actually a state of\/ $\mathcal{T}$, this is the case if $s$ is the source of $\nu$.
If $s$ is a transition $\zeta$ of\/ $\mathcal{T}$, this should be the case if $\nu$ is enabled in the source
of $\zeta$, and moreover $\nu$ and $\zeta$ are concurrent, in the sense that they stem from
different parallel components. A problem with this plan has been illustrated in \SSect{justice} by
the process $C|B$, with \plat{$B \stackrel{{\it def}}{=} c.B + b!.0$} and \plat{$C\stackrel{{\it def}}{=}c.C$}.
The $b!$-labelled transition $\nu$ is enabled in (or during) the transition \plat{$C|B\ar{c} C|B$} if this
transition stems from $C$, but not if it stems from $B$. However, our transition system\/
$\mathcal{S}$ fails to distinguish transitions based on the components from which they stem.

For this reason, we here define a different Kripke structure\/ $\mathcal{U}$ that makes the required distinctions.
The states of\/ $\mathcal{U}$ are the states $\T_{\rm\ABC}$ of\/ $\mathcal{T}$ together with the \emph{derivations} of
transitions of\/ $\mathcal{T}$ from the rules of \Tab{ABC}---the latter are the \emph{concrete transitions}
alluded to in \SSect{justice}. 
The transitions of\/ $\mathcal{U}$ are of the form
$P \rightarrow \chi$ and $\chi \rightarrow Q$, for\/ $\mathcal{T}$-states $P,Q$ and derivations $\chi$ corresponding
to a transition $P\ar{\alpha}Q$.
The Kripke structure\/ $\mathcal{S}$ can be
obtained from\/ $\mathcal{U}$ by consistently identifying multiple derivations of the same transition.

We start by giving a name to every derivation of an \ABC transition from the rules of \Tab{ABC}.
The unique derivation of the transition $\alpha.P \ar{\alpha} P$ using the rule \myref{Act} is
called $\shar{\alpha}P$. The derivation obtained by application of \myref{Comm} or \myref{Bro-c} on
the derivations $\chi$ and $\zeta$ of the premises of that rule is called $\chi|\zeta$. The
derivation obtained by application of \myref{Par-l} or \myref{Bro-l} on the derivation $\chi$ of the
(positive) premise of that rule, and using process $Q$ at the right of
$|$, is $\chi|Q$. In the same way, \myref{Par-r} and \myref{Bro-r} yield $P|\zeta$,
whereas \myref{Sum-l}, \myref{Sum-r}, \myref{Rel}, \myref{Res} and \myref{Rec} yield $\chi{+}Q$, $P{+}\chi$,
$\chi[f]$, $\chi\backslash c$ and \plat{$A\mathord:\chi$}.
For a derivation $\chi$ of a transition $P \ar{\alpha} Q$ write $\src(\chi):=P$, $\target(\chi):=Q$
and $\ell(\chi):=\alpha$.

It remains to define atomic propositions on\/ $\mathcal{U}$.
Following \Sect{properties} we have an atomic proposition $\alpha$ for each $\alpha\in Act$, and a
state $u$ of\/ $\mathcal{U}$ is labelled with $\alpha$ iff $u$ is a derivation of a transition $P\ar{\alpha}Q$.
Additionally, \SSect{justice} announced atomic propositions $\nu$ and $\en(\nu)$ for each
\emph{abstract transition} $\nu$; this is the subject of \App{enabling}.

Let $\widehat\cdot$ be the mapping from the states of\/ $\mathcal{U}$ to the states of\/ $\mathcal{S}$
given by $\widehat P = P$ for any process $P\in\T_{\rm\ABC}$, and $\widehat{\chi} = (P\ar{\alpha}Q)$
for any derivation $\chi$ of a transition $P\ar{\alpha}Q$ in\/ $\mathcal{T}$.
Then each state of\/ $\mathcal{S}$ is of the form $\widehat{u}$, and there is a transition $\widehat
u\rightarrow s$ in\/ $\mathcal{S}$ iff there is a transition $u\rightarrow v$ in\/ $\mathcal{U}$ with $\widehat v = s$.
A \emph{path} $\pi$ in\/ $\mathcal{U}$ is a finite or infinite sequence $u_0 u_1 u_2 \dots$ of
states of\/ $\mathcal{U}$ such that $u_i\rightarrow u_{i+1}$ for all $i$.
This amounts to an alternating sequence of processes $P\in\T_{\rm\ABC}$ and derivations $\chi$ of
transitions from\/ $\mathcal{T}$.
For any such path $\pi$ let $\widehat\pi := \widehat u_0 \widehat u_1 \widehat u_2 \dots$.
Then $\widehat\pi$ is a path in\/ $\mathcal{S}$, and furthermore any path in\/ $\mathcal{S}$ is of this form.
Both in\/ $\mathcal{U}$ and in\/ $\mathcal{S}$ we only consider paths that are infinite or end in a state of\/ $\mathcal{T}$.

\section{An Asymmetric Concurrency Relation between Transitions}\label{sec:concurrency}

\newcommand{\conc}{\mathrel{\mbox{$\smile\hspace{-1.15ex}\raisebox{2.5pt}{$\scriptscriptstyle\bullet$}$}}}
\newcommand{\nconc}{\mathrel{\mbox{$\not\smile\hspace{-1.15ex}\raisebox{2.5pt}{$\scriptscriptstyle\bullet$}$}}}
We define a \emph{concurrency} relation $\conc$ between the derivations of the outgoing transitions
of a process $P\in\T_{\rm\ABC}$. With $\chi \conc\zeta$ we mean that the possible occurrence of
$\chi$ is unaffected by the possible occurrence of $\zeta$. More precisely, $\chi$ and $\zeta$
need to be enabled in the state $P$, and $\chi \nconc\zeta$ indicates that the occurrence of $\zeta$
ends or interrupts the enabledness of $\chi$, whereas $\chi\conc\zeta$ indicates that $\chi$ remains
enabled during the execution of $\zeta$.

\begin{example}
Let $P$ be the process $A$ with \plat{$A\stackrel{{\it def}}{=}a.A + c.A$}, and let $\chi$ and
$\zeta$ be the derivations of the $a$- and $c$-transitions of $P$. Then $\chi \nconc\zeta$, because the
occurrence of $\zeta$ interrupts the enabledness of $\chi$, even though right after $\zeta$ has
occurred we again reach a state where $\chi$ is enabled.
\end{example}

\begin{example}
Let $P$ be the process $a.0|c.0$, and let $\chi$ and
$\zeta$ be the derivations of the $a$- and $c$-transitions. Then $\chi \conc\zeta$, because the
occurrence of $\zeta$ does not affect the (parallel) occurrence of $\chi$ in any way.
\end{example}

\begin{example}\label{ex:concurrency broadcast}
Let $P$ be the process $b!.0|(b?.0+c.0)$, and let $\chi$ and
$\zeta$ be the derivations of the $b!$- and $c$-transitions of $P$.
The broadcast $b!$ is in our view completely under the control of the left component; it will occur
regardless of whether the right component listens to it or not. It so happens that if $b!$ occurs in
state $P$, the right component will  listen to it, thereby disabling the possible occurrence of $c$.
For this reason we have $\chi \conc\zeta$ but $\zeta \nconc\chi$.
\end{example}

\begin{definition}\label{df:concurrency}\rm
\emph{Concurrency} is the smallest relation $\conc$ on derivations such that
\begin{itemize}
\vspace{-1ex}
\item $\chi|Q\conc P|\zeta$ and $P|\zeta\conc\chi|Q$ if $\src(\chi)=P$ and $\src(\zeta)=Q$,
\item $\chi|\varsigma \conc P|\zeta$ and $\varsigma|\chi \conc \zeta|P$ if $\src(\chi)=P$,
  $\src(\varsigma)=\src(\zeta)$ and $\ell(\varsigma)\in\B?$,
\item $\chi\conc\zeta$ implies
$\chi{+}P \mathbin{\conc} \zeta{+}P$,
$P{+}\chi \mathbin{\conc} P{+}\zeta$,
$\chi{|}P \mathbin{\conc} \zeta{|}P$ and
$P{|}\chi \mathbin{\conc} P{|}\zeta$,
\item  $\chi\mathbin{\conc}\zeta$ implies
$\chi{|}P \mathbin{\conc} \zeta{|}\xi$,
$\chi{|}\xi \mathbin{\conc} \zeta{|}P$,
$P{|}\chi \mathbin{\conc} \xi{|}\zeta$ and
$\xi{|}\chi \mathbin{\conc} P{|}\zeta$
if $P\mathbin=\src(\xi)$,
\item$\chi\conc\zeta$ implies
$\chi{|}\varsigma \conc \zeta{|}\xi$ and
$\varsigma{|}\chi \conc \xi{|}\zeta$ if $\src(\varsigma)=\src(\xi)$ and $\ell(\varsigma)\in\B?$,
\item $\chi\conc\zeta \wedge \varsigma \conc \xi$ implies
$\chi{|}\varsigma \conc \zeta{|}\xi$, and
\item $\chi\mathbin{\conc}\zeta$ implies
$\chi\backslash c \mathbin{\conc} \zeta\backslash c$,\,
$\chi[f] \mathbin{\conc} \zeta[f]$,\,
$A{:}\chi \mathbin{\conc} A{:}\zeta$ for any $c\mathbin\in\HC\!$, relabelling $f\!$ and $A\mathbin\in\A\!$,
\vspace{-1ex}
\end{itemize}
for arbitrary derivations $\chi$, $\zeta$, $\varsigma$, $\xi$,
and expressions $P,Q\in\T_{\rm\ABC}$, provided that the composed derivations exist.
We say that $\chi$ and $\zeta$ are \emph{concurrent} (in signs $\chi \smile \zeta$) if  $\chi\conc\zeta$ and $\zeta\conc\chi$.
\end{definition}

\begin{observation}\rm\label{obs:source}
The relation $\conc$ is irreflexive. Moreover, if $\chi\conc\zeta$ then $\src(\chi)=\src(\zeta)$.
\end{observation}
This follows by a straightforward induction on the definition of $\conc$.

\begin{example}\label{ex:concurrency}
Let \plat{$A\mathop{\stackrel{{\it def}}{=}}c.A$} and
\plat{$B\mathop{\stackrel{{\it def}}{=}}\bar c.B + (\tau.B+b!.0)$}.\vspace{-2pt}
Transition $A|B \mathbin{\ar{\tau}} A|B$ has 2 derivations:
$(A\mathord{:}{\shar{c}}A)|B\mathord:\big(({\shar{\bar c}}B){+}(\tau.B+b!.0)\big)$
and $A|B\mathord:\big(\bar c.B {+} ((\shar{\tau}B){+}b!.0)\big)$.
Only the latter is concurrent with \plat{$(A\mathord{:}{\shar{c}}A)|B$}
(using the first clause above).
\end{example}

\begin{example}\label{ex:concurrency 2}
One has $(((\shar{a}0)|c.0)+d.0)[f] \smile ((a.0|(\shar{c}0))+d.0)[f]$,
using the first, third and seventh clauses above.
Both are derivations of transitions with source $((a.0|c.0)+d.0)[f]$.
\end{example}

\begin{example}\label{ex:concurrency 3}
One has $((\shar{a}0)|c.0)|((\shar{\bar a}0)|\bar c.0) \smile (a.0|(\shar{c}0))|(\bar a.0|(\shar{\bar c}0)))$,
using the first and sixth clauses above.
Both are derivations of transitions with source $(a.0|c.0)|(\bar a.0|\bar c.0)$.
\end{example}

\begin{example}\label{ex:concurrency 4}
One has $((\shar{a}0)|c.0)|(\shar{\bar a}0) \smile (a.0|(\shar{c}0))|\bar a.0$,
using the first and fourth clauses above.
Both are derivations of transitions with source $(a.0|c.0)|\bar a.0$.
\end{example}

\begin{example}\label{ex:concurrency 5}
One has $\shar{b!}0\,|((\shar{b?}0)+c.0) \conc b!.0\,|(b?.0+(\shar{c}0))$,
using the second clause above.
Both are derivations of transitions with source $b!.0|(b?.0+c.0)$.
However, $b!.0\,|(b?.0+(\shar{c}0))\not\conc\shar{b!}0\,|((\shar{b?}0)+c.0)$.
See \Ex{concurrency broadcast} for motivation.
\end{example}

\begin{example}\label{ex:concurrency 6}
One has $((\shar{b!}0)|c.0)|((\shar{b?}0)+\bar c.0) \conc (b!.0|(\shar{c}0))|(b?0+(\shar{\bar c}0))$,
using the first and fifth clauses above.
Both are derivations of transitions with source $(b!.0|c.0)|(b!.0 + \bar c.0)$.
\end{example}

\newpage

\section{Enabling Abstract Transitions}\label{sec:enabling}

Below we define the concept of an \emph{abstract transition} as an equivalence class of
\emph{concrete} transitions, the latter being the derivations of ABC transitions from the rules of \Tab{ABC}.
The main idea is that a transition $\nu$ stemming from one side of a parallel composition yields only one
abstract transition of the parallel composition itself, regardless of the state of the other component,
and, if $\nu$ performs a broadcast action, regardless of whether the other component performs a
receive action synchronising with $\nu$.

Henceforth, $\nu$, $\nu_1$ and $\nu_2$, ranges over abstract transitions, $\chi$,
  $\zeta$, $\varsigma$ and $\xi$ over derivations, $P,Q$ over ABC expressions, and $u$, $v$ over
  states of $\mathcal{U}$, which are either derivations or ABC expressions.

\begin{definition}\label{df:equiv}\rm
Let $\equiv$ be the smallest equivalence relation on derivations $\chi$ with $\ell(\chi)\notin\B?$, satisfying
\begin{itemize}
\vspace{-1ex}
\item $\chi{|}P \equiv \chi{|}Q$ and $P{|}\chi \mathbin{\equiv} Q{|}\chi$,
\item $\chi{|}\varsigma \equiv \chi{|}P$ and $\varsigma{|}\chi \mathbin{\equiv} P{|}\chi$ if
  $\ell(\chi)\in\B!$ (and thus $\ell(\varsigma)\in\B?$),
\item $\chi{+}P \equiv \chi \equiv P{+}\chi$,~ $A\mathord:\chi\equiv\chi$ for $A\in\A$,
\item $\chi\equiv\zeta$ implies
$\chi\backslash c \mathbin{\equiv} \zeta\backslash c$,
$\chi[f] \mathbin{\equiv} \zeta[f]$,
$\chi{|}P \equiv \zeta{|}P$ and $P{|}\chi \mathbin{\equiv} P{|}\zeta$,
and moreover
\item $\chi\equiv\zeta \wedge \varsigma \equiv \xi$ implies
$\chi{|}\varsigma \equiv \zeta{|}\xi$,
\vspace{-1ex}
\end{itemize}
for arbitrary derivations $\chi$, $\zeta$, $\varsigma$, $\xi$,
and expressions $P,Q\in\T_{\rm\ABC}$, provided that the composed derivations exist.
An equivalence class $[\chi]_\equiv$ is called an \emph{abstract transition}; it 
can uniquely be denoted by leaving out $A\mathord:$, $P{+}$ and ${+}P$ and writing
$\zeta|\_$ for $\zeta|P$ or for $\zeta|\varsigma$ with $\ell(\zeta)\in\B!$---and likewise
$\_|\zeta$ for $P|\zeta$ or for $\varsigma|\zeta$ with $\ell(\zeta)\in\B!$---%
in all subexpressions of $\chi$. 
If $\nu=[\chi]_\equiv$ then the derivation $\chi$ is called a \emph{representative} of the abstract
transition $\nu$.
\end{definition}

By definition $\chi\equiv \zeta$ implies $\ell(\chi)=\ell(\zeta)$.
Setting $\ell([\chi]_\equiv):=\ell(\chi)$, we note that receive-actions are excluded, i.e.,
for any abstract transition $\nu$ we have $\ell(\nu)\notin\B?$.

\begin{observation}\label{obs:equiv_parallel}
If $\chi| u\equiv  v_1|v_2$ with $\ell(\chi)\not\in\B?$
then $\chi\equiv v_1$. 
As a consequence, since no derivation is related to a process, $\chi|u\not\equiv Q|\zeta$, for
any derivation~$\zeta$ and $Q\in\T_{\rm\ABC}$.
\end{observation}

\begin{observation}\label{obs:equiv_rename_restrict}
If $\chi[f]\equiv\zeta[f]$ or $\chi\backslash c\equiv\zeta\backslash c$ then $\chi\equiv\zeta$.
\end{observation}

The abstract transitions, with their labels, can be seen as the smallest set such that
\newcommand{\abstr}{abstract\ }
\newcommand{\transition}{tr.\ }
\newcommand{\transitions}{trs.\ }

\setlength\leftmargini{\labelwidth}
\begin{itemize}
\vspace{-1ex}
\item $\shar{\alpha}P$ is an \abstr \transition for $\alpha\in \B!\cup\HC\cup\{\tau\}$ and
  $P\in\T_{\rm \ABC}$, and $\ell(\shar{\alpha}P)=\alpha$,
\item if $\nu$ is an \abstr \transition then so are $\nu|\_$ and $\_|\nu$, with $\ell(\nu|\_)\mathbin=\ell(\_|\nu)\mathbin=\ell(\nu)$,
\item if $\nu_1$ and $\nu_2$ are \abstr \transitions with $\ell(\nu_1)\mathbin=\overline{\ell(\nu_2)}$
  then so is $\nu_1|\nu_2$, with $\ell(\nu_1|\nu_2)\mathbin=\tau$,
\item if $\nu$ is an \abstr \transition with $c\mathbin{\neq}\ell(\nu)\mathbin{\neq} \bar c$ then so is $\nu\backslash c$,
  with $\ell(\nu\backslash c)\mathbin=\ell(\nu)$, and
\item if $\nu$ is an \abstr \transition and $f$ a relabelling then so is $\nu[f]$, with $\ell(\nu[f])=f(\ell(\nu))$.
\vspace{-1ex}
\end{itemize}
Abstract transitions only reflect the syntactical structure of derivations; 
they do not take semantics into account. Hence, $\nu|\_\not=\_|\nu$ 
and $(\_|\nu)|\_\not=\_|(\nu|\_)$.

For each abstract transition $\nu$ we introduce atomic propositions $\nu$ and $\en(\nu)$.
The former says that $\nu$ occurs. It holds for a state $u$ of\/ $\mathcal{U}$ iff $u$ is a derivation
$\zeta$ such that $\nu=[\zeta]_\equiv$. 
The latter is defined by a case distinction on the type of state $u$.
An abstract transition $\nu$ is \emph{enabled} (denoted by $\en(\nu))$ in $P\mathbin\in\T_{\rm\ABC}$ iff $P=\src(\chi)$ for
a representative $\chi$ of $\nu$.
It is enabled in (or during) a derivation $\zeta$ iff $\nu$ has a representative $\chi$ with
$\chi \conc \zeta$.
As we shall see, in that case $\nu$ is also enabled in $\src(\zeta)$ as well as $\target(\zeta)$.
We write $u\models p$ if the atomic proposition $p$ holds in the state $u$ of\/ $\mathcal{U}$.

\begin{example}
The abstract transition $\nu=\_ | (\shar{c}0)$, with $c\mathbin\in\HC$ and representatives
$\chi_1=0|(\shar{c}0)$ and
$\chi_2=a.0+\big(e.0|(\shar{c}0)\big)$,
is enabled during the derivation $\zeta = a.0+\big((\shar{e}0)|c.0)\big)$
of the transition \plat{$a.0+(e.0|c.0) \ar{e} 0 | c.0$}.
This is the case because $\chi_2\smile \zeta$.
Accordingly, $\nu$ is also enabled in the source $\src(\zeta)= a.0+(e.0|c.0)$ as well as in the target
$\target(\zeta) = 0 | c.0$ of that transition.
This example would break down without the identification $\chi \equiv P{+}\chi$.
\end{example}

\begin{example}
Let $C \stackrel{{\it def}}= d.0 | e.0$. Then the abstract transition $\nu = \_|(\shar{d}0|\_)$ is enabled
during the derivation $\zeta = c.0 | C{:}(d.0 | \shar{e}0)$ of the transition $c.0 | C \ar{e} c.0 | (d.0 | 0)$.
Accordingly, it is enabled in the source $\src(\zeta)= c.0 | C$ as well as in the target
$\target(\zeta) = c.0 | (d.0 | 0)$ of that transition.
This example would break down without the identification $C\mathord:(d.0 | \_ )\equiv(d.0 |\_)$.
\end{example}

\begin{example}\label{ex:broadcast enabling}
Let $D \stackrel{{\it def}}= c.(b?.0 + e.D)$. In our view, the infinite path labelled $(ce)^\omega$
of the process $b!.0|D$ is unjust, because the output $b!$ is continuously enabled, yet never
taken. The idea is that the component $b!.0$ will perform this output regardless of whether the
other component is listening. For this reason, we need to formalise this output as a single abstract
transition $\nu$ such that the path labelled $(ce)^\omega$ satisfies $\mathbf{F}\mathbf{G}\,\en(\nu)$.
However, in the state $b!.0|D$---and during the execution of $b!.0|D{:}\shar{c}(b?.0 + e.D)$---the
derivation \plat{$(\shar{b!}0) | D$} is enabled, yet in the state $b!.0|(b?.0 + e.D)$ the derivation
\plat{$\shar{b!}0 | ((\shar{b?}0)+c.D)$} is enabled. In order to regard these two derivations as
representatives of the same abstract transition \plat{$(\shar{b!}0)|\_$} we employ the equivalence
$\chi|Q \equiv \chi|\varsigma$ when $\ell(\chi)\in\B!$.

Furthermore, during the execution of $b!.0|(b?.0 + \shar{e}D)$, a representative $\chi$ of
\plat{$(\shar{b!}0)|\_$} with $\src(\chi)=b!.0|(b?.0 + e.D)$ needs to be enabled as well.\vspace{2pt}
The only candidate is \plat{$\chi\,=\,\shar{b!}0 | ((\shar{b?}0)+c.D)$}, so
\plat{$\shar{b!}0 | ((\shar{b?}0)+c.D) \conc b!.0|(b?.0 + \shar{e}D)$}.
This is further motivation for the second clause in the \Def{concurrency} above.
\end{example}

\begin{lemma}\label{lem:compose parallel}
If $u\models\en(\nu)$ for a state $u$ of\/ $\mathcal{U}$
and an abstract transition $\nu$
then $u|v\models\en(\nu|\_)$ for any state $u|v$ of\/ $\mathcal{U}$.
\end{lemma}
\begin{proof}
We make case distinctions based on whether $u$ and $v$ are processes $P,Q$ or derivations.
\begin{itemize}
	\item Suppose $u=P\models\en(\nu)$. Then $P\mathbin=\src(\chi)$ for a representative $\chi$ of $\nu$.
		\begin{itemize}
			\item Let $v=Q\in\T_{\rm\ABC}$.
			In case $\ell(\chi)=b!$ and $Q\ar{b?}$, let $\varsigma$ be a derivation of a transition $Q\ar{b?}Q'$.
			Then there exists a derivation $\chi|\varsigma$, with $\src(\chi|\varsigma)=P|Q$, which is a representative of
			the abstract transition $\nu|\_$.
			In all other cases there exists a derivation $\chi|Q$, with $\src(\chi|Q)\mathbin= P|Q$, which is a
			representative of the abstract transition $\nu|\_$.
			Either way $P|Q\models\en(\nu|\_)$.			
			\item Now let $v=\xi$ be a derivation with $Q:=\src(\xi)$.
			In case $\ell(\chi)=b!$ and $Q\ar{b?}$, let $\varsigma$ be a derivation of a transition $Q\ar{b?}Q'$.
			Then there exists a derivation $\chi|\varsigma$, with
			$\chi|\varsigma \conc P|\xi$, which is a representative of
			the abstract transition $\nu|\_$.
			In all other cases there exists a derivation $\chi|Q$, with
			$\chi|Q \smile P|\xi$, which is a
			representative of the abstract transition $\nu|\_$.
			Either way $P|\xi\models\en(\nu|\_)$.
		\end{itemize}
	\item Suppose $u=\zeta\models\en(\nu)$. Then $\chi\conc\zeta$ for a representative $\chi$ of $\nu$.
		\begin{itemize}
			\item Let $v=Q\in\T_{\rm\ABC}$.
			In case $\ell(\chi)=b!$ and $Q\ar{b?}$, let $\varsigma$ be a derivation of a transition $Q\ar{b?}Q'$.
			Then there exists a derivation $\chi|\varsigma$, with $\chi|\varsigma \conc \zeta|Q$, which is a representative of
			the abstract transition $\nu|\_$.
			In all other cases there exists a derivation $\chi|Q$, with $\chi|Q \conc \zeta|Q$, which is a
			representative of the abstract transition $\nu|\_$.
			Either way $\zeta|Q\models\en(\nu|\_)$.
			
			\item Now let $v=\xi$ be a derivation with $Q:=\src(\xi)$.
			In case $\ell(\chi)=b!$ and $Q\ar{b?}$, let $\varsigma$ be a derivation of a transition $Q\ar{b?}Q'$.
			Then there exists a derivation $\chi|\varsigma$, with
			$\chi|\varsigma \conc \zeta|\xi$, which is a representative of
			the abstract transition $\nu|\_$.
			In all other cases there exists a derivation $\chi|Q$, with
			$\chi|Q \conc \zeta|\xi$, which is a
			representative of the abstract transition $\nu|\_$.
			Either way $\zeta|\xi\models\en(\nu|\_)$.
                        \qed
		\end{itemize}
	\end{itemize}
\end{proof}

\begin{lemma}\label{lem:compose parallel 2}
Let $u_i\mathbin{\models}\en(\nu_i)$ for $i\mathord=1,2$ with $\ell(\nu_1)=\overline{\ell(\nu_2)}\in
\HC$. Then $u_1|u_2\mathbin{\models}\en(\nu_1|\nu_2)$, provided $u_1|u_2$ is a state of\/ $\mathcal{U}$.
\end{lemma}
\begin{proof}
We make case distinctions based on whether $u_1$ and $u_2$ are processes or derivations.

Suppose $u_i=P_i\models\en(\nu_i)$ for $i\mathord=1,2$. Then $P_i\mathbin=\src(\chi_i)$ for
representatives $\chi_i$ of $\nu_i$.
Now $\src(\chi_1|\chi_2)\mathbin= P_1|P_2$ and $\chi_1|\chi_2$ is a representative of
$\nu_1|\nu_2$. So $P_1|P_2\models\en(\nu_1|\nu_2)$.

Suppose $u_i\mathbin=\zeta_i\mathbin{\models}\en(\nu_i)$ for $i\mathord=1,2$. Then $\chi_i\mathbin{\conc}\zeta_i$ for
representatives $\chi_i$ of $\nu_i$.\\
So $\chi_1|\chi_2 \conc\zeta_1|\zeta_2$. Moreover, $\chi_1|\chi_2$ is a representative of $\nu_1|\nu_2$,
and thus $\zeta_1|\zeta_2 \models\en(\nu_1|\nu_2)$.

Suppose $P_1\models\en(\nu_1)$ and $\zeta_2\models\en(\nu_2)$. Then $P_1\mathbin=\src(\chi_1)$
and $\chi_2\conc\zeta_2$ for representatives $\chi_i$ of $\nu_i$.
Now $\chi_1|\chi_2 \conc P_1|\zeta_2$ and $\chi_1|\chi_2$ is a representative of $\nu_1|\nu_2$.
So $P_1|\zeta_2 \models\en(\nu_1|\nu_2)$.

The remaining case follows by symmetry.
\qed\end{proof}

\begin{lemma}\label{lem:compose restriction}
If $u\models\en(\nu)$, $c\mathbin\in\HC$ and $c\mathbin{\neq} \ell(\nu)\mathbin{\neq} \bar c$
then $u\backslash c\models\en(\nu\backslash c)$.
\end{lemma}
\begin{proof}
We make a case distinction based on whether $u$ is processes $P$ or a derivations $\zeta$.

Suppose $P\models\en(\nu)$. Then $P\mathbin=\src(\chi)$ for a representative $\chi$ of $\nu$.
Now $\src(\chi\backslash c)\mathbin= P\backslash c$ and $\chi\backslash c$ is a representative of
$\nu\backslash c$. So $P\backslash c\models\en(\nu\backslash c)$.

Suppose $\zeta\models\en(\nu)$. Then $\chi\conc\zeta$ for a representative $\chi$ of $\nu$.
So $\chi\backslash c \conc\zeta\backslash c$. Moreover, $\chi\backslash c$ is a representative of
$\nu\backslash c$, and thus $\zeta\backslash c\models\en(\nu\backslash c)$.
\qed\end{proof}

\begin{lemma}\label{lem:compose relabelling}
If $u\models\en(\nu)$ then $u[f]\models\en(\nu[f])$.
\end{lemma}
\begin{proof}
Similar to the proof of Lemma~\ref{lem:compose restriction}.
\qed\end{proof}

\begin{proposition}\label{prop:target}
If an abstract transition $\nu$ is enabled during a derivation $\zeta$
then $\nu$ is also enabled in $\src(\zeta)$ as well as $\target(\zeta)$.
\end{proposition}

\begin{proof}
If $\nu$ is enabled during $\zeta$ then
there is a representative $\chi$ of $\nu$ such that $\chi\conc\zeta$.
By \Obs{source} $\src(\chi)=\src(\zeta)$, so $\nu$ is also enabled in $\src(\zeta)$.

For the other statement, we apply structural induction on $\chi$.

Let $\chi=\shar{\alpha}\chi'$. By Definition~\ref{df:concurrency} there is 
no derivation $\zeta$ with $\chi\conc \zeta$, which is a contradiction to our
assumptions.

Let $\chi=A{:}\chi'$. Since $\chi\conc \zeta$, $\zeta$ has the form $A{:}\zeta'$ with $\chi'\conc\zeta'$.
As $\nu=[\chi]_\equiv=[\chi']_\equiv$, $\nu$ is also enabled during $\zeta'$, and 
by induction $\nu$ is enabled in $\target(\zeta')=\target(\zeta)$.

Let $\chi=\chi'{+}P$. Since $\chi\conc \zeta$, $\zeta$ has the form $\zeta'{+}P$ with $\chi'\conc\zeta'$.
As $\nu=[\chi]_\equiv=[\chi']_\equiv$, $\nu$ is also enabled during $\zeta'$, and 
by induction $\nu$ is enabled in $\target(\zeta')=\target(\zeta)$.

The case $\chi=P{+}\chi'$  follows by symmetry.

Let $\chi=\chi'\backslash c$.  Since $\chi\conc \zeta$, $\zeta$ has the form $\zeta'\backslash c$
with $\chi'\conc\zeta'$.
So $\nu':=[\chi']_\equiv$ is enabled during $\zeta'$, and
by induction $\nu'$ is enabled in $\target(\zeta')$. Moreover, $\ell(\nu')=\ell(\chi')\neq c,\bar c$.
Consequently, using \Lem{compose restriction},
$\nu=\nu'\backslash c$ is enabled in $\target(\zeta')\backslash c = \target(\zeta)$.

Let $\chi=\chi'[f]$.  Since $\chi\conc \zeta$, $\zeta$ has the form $\zeta'[f]$ with $\chi'\conc\zeta'$.
So $\nu':=[\chi']_\equiv$ is enabled during $\zeta'$, and
by induction $\nu'$ is enabled in $\target(\zeta')$.
Consequently, using \Lem{compose relabelling}, $\nu=\nu'[f]$ is enabled in $\target(\zeta')[f] = \target(\zeta)$.

Let $\chi=\chi_1|P$. Since $\chi\conc \zeta$, \Def{concurrency} offers three possibilities for $\zeta$:
\begin{itemize}
\vspace{-1ex}
\item
Suppose that $\zeta$ has the form $Q|\zeta_2$ with $\src(\chi_1)=Q$ and
$\src(\zeta_2)=P$. Then $\nu_1:=[\chi_1]_\equiv$ is enabled in $Q$.
Hence, by \Lem{compose parallel}, $\nu=([\chi_1]_\equiv)|\_$ is enabled in
$P|\target(\zeta_2)=\target(\zeta)$.
\item
Suppose that $\zeta$ has the form $\zeta_1|P$ with $\chi_1\conc\zeta_1$.
Then $\nu_1:=[\chi_1]_\equiv$ is enabled during $\zeta_1$, and
by induction $\nu_1$ is enabled in $\target(\zeta_1)$.
Consequently, using \Lem{compose parallel}, $\nu=\nu_1|\_$ is enabled in $\target(\zeta_1)|P = \target(\zeta)$.
\item
Suppose that $\zeta$ has the form $\zeta_1|\zeta_2$ with
$\chi_1\conc\zeta_1$ and $P=\src(\zeta_2)$.
Then $\nu_1:=[\chi_1]_\equiv$ is enabled during $\zeta_1$, and
by induction $\nu_1$ is enabled in $\target(\zeta_1)$.
Consequently, using \Lem{compose parallel}, $\nu=\nu_1|\_$ is enabled in $\target(\zeta_1)|\target(\zeta_1) = \target(\zeta)$.
\vspace{-1ex}
\end{itemize}
The case $\chi= P|\chi_2$ follows by symmetry.

Let $\chi=\chi_1|\chi_2$ with $\ell(\chi)=\tau$. Then $\ell(\chi_1)=\overline{\ell(\chi_2)}\in\HC$.
Since $\chi\conc \zeta$, \Def{concurrency} offers three possibilities for $\zeta$:
\vspace{-1ex}
\begin{itemize}
\item
Suppose $\zeta$ has the form $\zeta_1|P$ with $\chi_1\conc\zeta_1$ and $P=\src(\chi_2)$.
Then $\nu_1:=[\chi_1]_\equiv$ is enabled during $\zeta_1$, and
by induction $\nu_1$ is enabled in $\target(\zeta_1)$.
Consequently, using \Lem{compose parallel}, $\nu=\nu_1|\_$ is enabled in $\target(\zeta_1)|P = \target(\zeta)$.
\item
The case that $\zeta$ has the form $P|\zeta_2$ with $\chi_2\conc\zeta_2$ and $P=\src(\chi_1)$
follows by symmetry.
\item
Suppose $\zeta$ has the form $\zeta_1|\zeta_2$ with $\chi_i\conc\zeta_i$ for $i=1,2$.
Then $\nu_i:=[\chi_i]_\equiv$ is enabled during $\zeta_i$, and
by induction $\nu_i$ is enabled in $\target(\zeta_i)$.
Consequently, using \Lem{compose parallel 2},
$\nu=\nu_1|\nu_2$ is enabled in $\target(\zeta_1)|\target(\zeta_2) = \target(\zeta)$.
\vspace{-1ex}
\end{itemize}
Let $\chi=\chi_1|\chi_2$ with $\ell(\chi)\neq\tau$.  Then $\ell(\chi)\in\B!$, since the case
$\ell(\chi) = \ell(\nu)\in\B?$ cannot occur. So $\ell(\chi_1)=b!$ and $\ell(\chi_2)=b?$ for some
$b\in\B$, or vice versa. W.l.o.g.\ we assume the first of these cases.
Since $\chi\conc \zeta$, \Def{concurrency} offers five possibilities for $\zeta$:
\begin{itemize}
\vspace{-1ex}
\item
Suppose $\zeta$ has the form $P|\zeta_2$ with $P=\src(\chi_1)$ and $\src(\chi_2)=\src(\zeta_2)$.
Then $\nu_1:=[\chi_1]_\equiv$ is enabled in $P$.
Hence, by \Lem{compose parallel}, $\nu=([\chi_1]_\equiv)|\_$ is enabled in
$P|\target(\zeta_2)=\target(\zeta)$.
\item
The possibility that $\zeta \mathbin= P|\zeta_2$ with $\chi_2\mathbin{\conc}\zeta_2$ and $P\mathbin=\src(\chi_1)$
is a special case of the last~one.
\item
Suppose $\zeta$ has the form $\zeta_1|P$ with $\chi_1\conc\zeta_1$ and $P=\src(\chi_2)$.
Then $\nu_1:=[\chi_1]_\equiv$ is enabled during $\zeta_1$, and
by induction $\nu_1$ is enabled in $\target(\zeta_1)$.
Consequently, using \Lem{compose parallel}, $\nu=\nu_1|\_$ is enabled in $\target(\zeta_1)|P = \target(\zeta)$.

\item
Suppose $\zeta$ has the form $\zeta_1|\zeta_2$ with $\chi_1\conc\zeta_1$ and $\src(\chi_2)=\src(\zeta_2)$.
Then $\nu_1:=[\chi_1]_\equiv$ is enabled during $\zeta_1$, and
by induction $\nu_1$ is enabled in $\target(\zeta_1)$.
Consequently, using \Lem{compose parallel}, $\nu=\nu_1|\_$ is enabled in
$\target(\zeta_1)|\target(\zeta_2) = \target(\zeta)$.

\item
The possibility $\zeta = \zeta_1|\zeta_2$ with $\chi_i\conc\zeta_i$ for $i=1,2$ is a special case of the previous one.
\qed
\end{itemize}
\end{proof}

\begin{lemma}\label{lm:conc_and_equiv}
For derivations $\chi$ and $\zeta$, $\chi\conc\zeta$ implies $\chi\not\equiv\zeta$.
\end{lemma}

\begin{proof}
In case $\ell(\chi) \in\B?$ the statement is trivial by \Def{equiv}; so assume $\ell(\chi) \not\in\B?$.
We apply structural induction on $\chi$.

Let $\chi=\shar{\alpha}\chi'$. By Definition~\ref{df:concurrency} there is 
no derivation $\zeta$ with $\chi\conc \zeta$, which is a contradiction to the antecedent.

Let $\chi=A{:}\chi'$. Since $\chi\conc \zeta$, $\zeta$ has the form $A{:}\zeta'$ with $\chi'\conc\zeta'$.
Assume $A{:}\chi'\equiv A{:}\zeta'$. Then, by Definition~\ref{df:equiv}, $\chi'\equiv
A{:}\chi'\equiv A{:}\zeta' \equiv \zeta'$, a contradiction to the induction hypothesis.

Let $\chi=\chi'{+}P$. Since $\chi\conc \zeta$, $\zeta$ has the form $\zeta'{+}P$ with $\chi'\conc\zeta'$.
Assume $\chi'{+}P\equiv \zeta'{+}P$. Then, by Definition~\ref{df:equiv},
$\chi'\equiv \chi'{+}P\equiv \zeta'{+}P \equiv \zeta'$, a contradiction 
to the induction hypothesis.

The case $\chi=P{+}\chi'$  follows by symmetry.

Let $\chi=\chi'\backslash c$.  Since $\chi\conc \zeta$, $\zeta$ has the form $\zeta'\backslash c$
with $\chi'\conc\zeta'$.
Assume $\chi'\backslash c\equiv \zeta'\backslash c$. Then, by \Obs{equiv_rename_restrict}, $\chi'\equiv\zeta'$, a contradiction 
to the induction hypothesis.

Let $\chi=\chi'[f]$.  Since $\chi\conc \zeta$, $\zeta$ has the form $\zeta'[f]$ with $\chi'\conc\zeta'$.
Assume $\chi'[f]\equiv \zeta'[f]$. Then, by \Obs{equiv_rename_restrict}, $\chi'\equiv\zeta'$, a contradiction 
to the induction hypothesis.

Let $\chi=\chi_1|Q$. Note that $\ell(\chi_1)=\ell(\chi)\not\in\B?$.
Since $\chi\conc \zeta$, \Def{concurrency} offers three possibilities for $\zeta$:
\begin{itemize}
\vspace{-1ex}
\item
Suppose that $\zeta$ has the form $P|\zeta_2$.
 By \Obs{equiv_parallel}, $\chi_1|Q\not\equiv P|\zeta_2$.

\item
Suppose that $\zeta$ has the form
$\zeta_1|u$ with $\chi_1\conc\zeta_1$ (combining the cases $\zeta_1|Q$ and $\zeta_1|\zeta_2$).
Assume $\chi_1|Q\equiv \zeta_1|u$. Then, by \Obs{equiv_parallel}, $\chi_1\equiv\zeta_1$, a contradiction 
to the induction hypothesis.
\vspace{-1ex}
\end{itemize}
The case $\chi= P|\chi_2$ follows by symmetry.

Let $\chi=\chi_1|\chi_2$ .
As $\ell(\chi)\not\in\B?$ either $\ell(\chi_1)\not\in\B?$ or $\ell(\chi_2)\not\in\B?$.
W.l.o.g.\ we assume the first.
Since $\chi\conc \zeta$, \Def{concurrency} offers seven possibilities for $\zeta$, 
which can be summarised by the following 4 cases.
\begin{itemize}
\vspace{-1ex}
\item 
Suppose $\zeta$ has the form $P|\zeta_2$. Then by \Obs{equiv_parallel} $\chi_1|\chi_2\not\equiv P|\zeta_2$.

\item The case where $\zeta$ has the form $\zeta_1|P$ with $\src(\chi_2)=P$, $\src(\chi_1)=\src(\zeta_1)$ and $\ell(\chi_1)\in\B?$ cannot occur since $\ell(\chi_1)\not\in\B?$.

\item The case where $\zeta$ has the form $\zeta_1|\zeta_2$ with $\chi_2\conc\zeta_2$, $\src(\chi_1)=\src(\zeta_1)$ and $\ell(\chi_1)\in\B?$
cannot occur either since $\ell(\chi_1)\not\in\B?$.

\item
Finally, suppose $\zeta$ has the form $\zeta_1|u$ with $\chi_1\conc\zeta_1$.
Assume $\chi_1|\chi_2\equiv \zeta_1|u$. Then, by \Obs{equiv_parallel}, $\chi_1\equiv\zeta_1$, a contradiction 
to the induction hypothesis.
\qed
\end{itemize}
\end{proof}

\section{Proof of \Thm{just path}}\label{sec:proof}

\Thm{just path} makes a connection between the $Y\!$-just paths in\/ $\mathcal{T}$ (or equivalently\/
$\mathcal{S}$) and a weak fairness property for paths in\/ $\mathcal{U}$.
To establish its ``$\Leftarrow$''-direction, we introduce two intermediate concepts:
the $Y\!$-just paths in\/ $\mathcal{U}$, for $Y\subseteq\HC$, and the $\nu$-enabled paths in\/
$\mathcal{U}$, for abstract transitions $\nu$. For the ``$\Rightarrow$''-direction we introduce one
intermediate concept: the $\nu$-enabled paths in\/ $\mathcal{S}$.

\subsection{The ``$\Leftarrow$''-direction}

\begin{observation}\label{obs:decompose parallel}\rm
For a derivation $\chi$ with $\src(\chi)= P_1|P_2$ we have either that
\begin{itemize}
\vspace{-1ex}
\item $\chi$ has the form $\chi_1|P_2$ with $\src(\chi_1)=P_1$ and $\target(\chi)=\target(\chi_1)|P_2$, or
\item $\chi$ has the form $\chi_1|\chi_2$ with $\src(\chi_i)\mathbin=P_i$ for $i\mathord=1,2$\\ and $\target(\chi)\mathbin=\target(\chi_1)|\target(\chi_2)$,~or
\item $\chi$ has the form $P_1|\chi_2$ with $\src(\chi_2)=P_2$ and $\target(\chi)=P_1|\target(\chi_1)$.
\vspace{-1ex}
\end{itemize}
\end{observation}
Hence all processes and derivations on a path $\pi$ starting from a state $u_1|u_2$ of\/ $\mathcal{U}$ have the form $\_|\_$.
Let $\pi_1$ be the sequence of left- and $\pi_2$ the sequence of right-components of these
processes and derivations, after (finite or infinite) subsequences of repeated elements are contracted
to single elements. Then $\pi_i$ is a path of $u_i$ ($i\mathord=1,2$) and together they constitute
the \emph{decomposition} of $\pi$, denoted $\pi\Rrightarrow\pi_1|\pi_2$. 

\begin{observation}\label{obs:decompose restriction}\rm
If $\src(\chi)\mathbin= P\backslash c$ then $\chi \mathbin= \chi'\backslash c$ with $\src(\chi')\mathbin=P$ and
$\target(\chi)\mathbin=\target(\chi') \backslash c$.%
\end{observation}
Hence all processes and derivations on a path $\pi$ of a state $u\backslash c$ in\/ $\mathcal{U}$ have the form $\_\backslash c$.
Let $\pi'$ be the sequence obtained from $\pi$ by stripping off these outermost occurrences of $\backslash c$.
Then $\pi'$ is a path of $u$, called the \emph{decomposition} of $\pi$, denoted $\pi\Rrightarrow\pi'\backslash c$.
In the same way one defines the decomposition of a path $\pi$ of a state $u[f]$; notation $\pi\Rrightarrow\pi'[f]$.

\begin{observation}\label{obs:decompose projection}\rm
Let $\pi$ be a path in\/ $\mathcal{U}$. Then $\pi\Rrightarrow\pi_1|\pi_2$ implies $\widehat\pi\in\widehat\pi_1|\widehat\pi_2$.
Likewise, if $\pi\Rrightarrow\pi'\backslash c$ or $\pi\Rrightarrow\pi'[f]$ then $\widehat\pi'$ is a
decomposition of $\widehat\pi$.
\end{observation}
Although in\/ $\mathcal{S}$ the decomposition of a path from $P|Q$ need not be unique, in\/
$\mathcal{U}$ it is. Armed with these definitions of decompositions of paths in\/ $\mathcal{U}$, we
define $Y\!$-justness, for $Y\subseteq\HC$, on the paths of\/ $\mathcal{U}$ in the exact same way as
on\/ $\mathcal{S}$ (see \Def{just path}).

\begin{proposition}\label{prop:projection}
If a path $\pi$ in\/ $\mathcal{U}$ is $Y\!$-just, for $Y\mathbin\subseteq\HC\!$, then so is the path
$\widehat\pi$ in\/ $\mathcal{S}$.
\end{proposition}

\newcommand{\ujustn}{just${}_{\,\mathcal{U}}$}
\newcommand{\ujust}{just${}_{\,\mathcal{U}}$\xspace}
\begin{proof}
  Define a path in\/ $\mathcal{S}$ to be $Y\!$-\emph{\ujustn}, for $Y\mathbin\subseteq\HC$, if it
  has the form $\widehat\pi$ for a $Y\!$-just path $\pi$ in\/ $\mathcal{U}$.
  We show that the family of predicates $Y\!$-justness${}_{\,\mathcal{U}}$
  satisfies the five requirements of \Def{just path}.
  \begin{itemize}
\vspace{-1ex}
  \item A finite $Y\!$-\ujust path $\widehat\pi$, with $\pi$ a just path in\/ $\mathcal{U}$, ends in some state $\widehat P$.
  Since $\widehat{P}=P$, $\pi$ ends in the same state.
    Hence that state admits actions from $Y\cup\B?$ only.
  \item Let $\pi$ be a $Y\!$-just path in\/ $\mathcal{U}$, so that $\widehat\pi$, and hence $\pi$,
    starts from a process $P|Q$. Then $\pi \Rrightarrow \pi_1|\pi_2$
    for an $X$-just path $\pi_1$ of $P$ and a $Z$-just path $\pi_2$ of $Q$ such that $Y\supseteq X\mathord\cup Z$
    and $X\mathord\cap \bar{Z}\mathbin=\emptyset$.
    By \Obs{decompose projection} $\widehat\pi\in\widehat\pi_1|\widehat\pi_2$, where
    $\widehat\pi_1$ is $X$-\ujust and $\widehat\pi_2$ is $Z$-\ujust.
  \item Let $\pi$ be a $Y\!$-just path in\/ $\mathcal{U}$ starting from a process $P\backslash c$.
    Then $\pi\Rrightarrow\pi'\backslash c$ for a $Y\mathord\cup\{c,\bar c\}$-just path $\pi'$ of $P$.
    By \Obs{decompose projection} $\widehat\pi'$ is a decomposition of $\widehat\pi$, where
    $\widehat\pi'$ is $Y\mathord\cup\{c,\bar c\}$-\ujust.
  \item The case that $\widehat\pi$ is a path of $P[f]$ proceeds in exactly the same way.
  \item Each suffix of $\widehat\pi$ has the form $\widehat\pi'$ for $\pi'$ a suffix of $\pi$.
    So if $\widehat\pi$ is $Y\!$-\ujust because $\pi$ is $Y\!$-just then $\pi'$ must be $Y\!$-just,
    and hence $\widehat\pi$ is $Y\!$-\ujust.
\vspace{-1ex}
  \end{itemize}
  Since $Y\!$-justness is the largest family of predicates on paths in\/ $\mathcal{S}$ that satisfies
  those requirements, $Y\!$-justness${}_{\,\mathcal{U}}$ of paths in\/ $\mathcal{S}$ implies $Y\!$-justness of paths in\/ $\mathcal{S}$.
\qed\end{proof}

\begin{definition}\label{df:nu-enabled path}\rm
\emph{$\nu$-enabledness}, for $\nu$ an abstract transition, is the smallest family of predicates on
the paths in\/ $\mathcal{U}$ such that
\begin{itemize}
\vspace{-1ex}
\item a finite path is $\nu$-enabled if its last state $Q\in\T_{\rm\ABC}$ enables $\nu$, i.e.\ $Q\models\en(\nu)$;
\item a path $\pi \Rrightarrow \pi_1|\pi_2$ is $\nu$-enabled if either $\nu$ has the form $\nu_1|\_$ and $\pi_1$ is 
  $\nu_1$-enabled, or $\nu=\_|\nu_2$ and $\pi_2$ is $\nu_2$-enabled, or $\nu=\nu_1|\nu_2$ and
  $\pi_i$ is $\nu_i$-enabled for $i=1,2$;
\item a path $\pi\Rrightarrow\pi'\backslash c$ is $\nu$-enabled if $\nu$ has the form $\nu'\backslash c$
  and $\pi'$ is $\nu'$-enabled;
\item a path $\pi\Rrightarrow\pi'[f]$ is $\nu$-enabled if $\nu$ has the form $\nu'[f]$ and
  $\pi'$ is $\nu'$-enabled;
\item and a path is $\nu$-enabled if it has a suffix that is $\nu$-enabled. 
\end{itemize}
\end{definition}

\begin{proposition}\label{prop:enabled just}
Let $\pi$ be a path in\/ $\mathcal{U}$ and $Y\subseteq\HC$. 
If, for all abstract transitions $\nu$ with $\ell(\nu)\not\in Y$, $\pi$ is not $\nu$-enabled, then $\pi$ is $Y\!$-just.
\end{proposition}

\newcommand{\ejustn}{just${}_{en}$}
\newcommand{\ejust}{just${}_{en}$\xspace}
\begin{proof}
  Define a path $\pi$ in\/ $\mathcal{U}$ to be $Y\!$-\emph{\ejustn}, for $Y\subseteq\HC$,
  if it is $\nu$-enabled for no abstract transition $\nu$ with $\ell(\nu)\not\in Y$.
  Note that if $\pi$ is $Y\!$-\ejust, it is also $Y'\!$-\ejust for any $Y\subseteq Y'\subseteq\HC$.
  We show that the family of predicates $Y\!$-justness${}_{en}$, for $Y\mathbin\subseteq \HC\!$,
  satisfies the five requirements of \Def{just path}.
  \begin{itemize}
\vspace{-1ex}
  \item Let $\pi$ be a finite $Y\!$-\ejust path. Suppose the last state $Q$ of $\pi$ admits an action $\alpha\notin Y\cup\B?$.
  Then $Q\models \en(\nu)$ for an abstract transition $\nu$ with $\ell(\nu)\mathop=\alpha\mathbin{\not\in} Y\!\mathbin{\subseteq} Y\mathord{\cup}\B?$.
  So $\pi$ is $\nu$-enabled, contradicting the $Y\!$-justness${}_{en}$ of $\pi$.
 
  \item Suppose $\pi$ is a $Y\!$-\ejust path of a process $P|Q$.
  Then $Y$ includes all labels of abstract transitions $\nu$ for which $\pi$ is $\nu$-enabled.
  By \Obs{decompose parallel} there are paths $\pi_i$ for $i\mathord=1,2$ with $\pi\Rrightarrow\pi_1|\pi_2$.
  Let $X$ be the set of labels of abstract transitions $\nu$ for which $\pi_1$ is $\nu$-enabled,
  and let $Z$ be the set of labels of abstract transitions $\nu$ for which $\pi_2$ is $\nu$-enabled.
  If $\pi_1$ is $\nu$-enabled then $\pi$ is $\nu|\_$-enabled by \Def{nu-enabled path}. Since
  $\ell(\nu|\_)=\ell(\nu)$ this implies that $X\subseteq Y$.
  In the same way it follows that $Z\subseteq Y$.

  Now suppose that $X\cap \bar{Z} \neq\emptyset$.
  Then $\pi_i$ is $\nu_i$-enabled, for $i\mathord=1,2$, for abstract transitions $\nu_i$ with
  $\ell(\nu_1)=c \in \HC$ and $\ell(\nu_2)=\bar c$.
  So by \Def{nu-enabled path} $\pi$ is $\nu_1|\nu_2$-enabled, in contradiction with $\ell(\nu_1|\nu_2)=\tau\not\in Y\subseteq \HC$.
  We therefore conclude that $X\cap \bar{Z} = \emptyset$.

  By definition, $\pi_1$ is $X$-\ejust and $\pi_2$ is $Z$-\ejust.

  \item Suppose $\pi$ is a $Y\!$-\ejust path of a process $P\backslash c$.
  Then $Y$ includes all labels of abstract transitions $\nu$ for which $\pi$ is $\nu$-enabled.
  By \Obs{decompose restriction} there is a path $\pi'$ with $\pi\Rrightarrow\pi'\backslash c$.
  Let $X$ be the set of labels of abstract transitions $\nu$ for which $\pi'$ is $\nu$-enabled.
  If $\pi'$ is $\nu$-enabled and $c\neq \ell(\nu)\neq \bar c$ then $\pi$ is $\nu\backslash c$-enabled by \Def{nu-enabled path}.
  Since $\ell(\nu\backslash c)=\ell(\nu)$ this implies that $X\setminus\{c,\bar c\} \subseteq Y$.
  It follows that $\pi'$ is $X$-\ejust, and hence $Y\mathord\cup\{c,\bar c\}$-\ejust.

  \item Suppose $\pi$ is a $Y\!$-\ejust path of a process $P[f]$.
  Then $Y$ includes all labels of abstract transitions $\nu$ for which $\pi$ is $\nu$-enabled.
  By the remark after \Obs{decompose restriction}, there is a path $\pi'$ with $\pi\Rrightarrow\pi'[f]$.
  Let $X$ be the set of labels of abstract transitions $\nu$ for which $\pi'$ is $\nu$-enabled.
  If $\pi'$ is $\nu$-enabled then $\pi$ is $\nu[f]$-enabled by \Def{nu-enabled path}.
  Since $\ell(\nu[f])=f(\ell(\nu))$ this implies that $f(X) \subseteq Y$.
  It follows that $\pi'$ is $X$-\ejust, and hence $f^{-1}(Y)$-\ejust.

  \item Suppose $\pi'$ is a suffix of an $Y\!$-\ejust path $\pi$.
  Then $Y$ includes all labels of abstract transitions $\nu$ for which $\pi$ is $\nu$-enabled.
  By the last clause of \Def{nu-enabled path}, $Y$ thereby includes all labels of abstract transitions
  $\nu$ for which $\pi'$ is $\nu$-enabled. Hence $\pi'$ is $Y\!$-\ejust.
  \end{itemize}
  Since $Y\!$-justness is the largest family of
  predicates that satisfies those requirements, $Y\!$-justness${}_{en}$ implies $Y\!$-justness.
\qed\end{proof}
Henceforth, we write $\pi\models\phi$ if the LTL formula $\phi$ holds for the path $\pi$ in\/
$\mathcal{U}$, that is, if $\pi$ satisfies $\phi$. 
Note that a finite path satisfies $\mathbf{FG}\,\phi$ if
$\phi$ holds in its last state.
\begin{proposition}\label{prop:nu-enabled}
If $\pi$ is $\nu$-enabled then $\pi\models \mathbf{FG}\,\en(\nu)$.
\end{proposition}

\begin{proof}
  We apply induction on $\nu$-enabledness of a path $\pi$ in\/ $\mathcal{U}$, using the
  five clauses of \Def{nu-enabled path}.

  Suppose $\pi$ is $\nu$-enabled because it is finite and its last state $Q\in\T_{\rm\ABC}$ enables $\nu$.
  Then $\pi \models \mathbf{FG}\,\en(\nu)$.

  Suppose $\pi \Rrightarrow \pi_1|\pi_2$ is $\nu$-enabled because $\nu=\nu_1|\_$ and $\pi_1$ is $\nu_1$-enabled.
  By induction $\pi_1 \models \mathbf{FG}\,\en(\nu_1)$.
  Let $\pi_1'= u_0 u_1 u_2 \dots$ be a suffix of $\pi_1$ with $\pi_1' \models \mathbf{G}\,\en(\nu_1)$.
  Then, for all $i\geq 0$, $u_i\models\en(\nu_1)$ and thus $u_i|v\models\en(\nu_1|\_)$ for any state\pagebreak[2] $u_i|v$
  of\/ $\mathcal{U}$ by Lemma~\ref{lem:compose parallel}.
  Therefore $\pi\models \mathbf{FG}\,\en(\nu)$.
  The case that $\nu=\_|\nu_2$ and $\pi_2$ is $\nu_2$-enabled goes likewise.

  Suppose $\pi \Rrightarrow \pi_1|\pi_2$ is $\nu$-enabled because $\nu=\nu_1|\nu_2$ and $\pi_i$
  is $\nu_i$-enabled for $i=1,2$. Then $\ell(\nu_1)=\overline{\ell(\nu_2)}\in\HC$.
  By induction $\pi_i \models \mathbf{FG}\,\en(\nu_i)$.
  Let $\pi_1'= u_0 u_1 u_2 \dots$ and $\pi_2'= v_0 v_1 v_2 \dots$ be (finite or infinite) suffixes of $\pi_1$ and $\pi_2$
  with $\pi_i' \models \mathbf{G}\,\en(\nu_i)$ for $i=1,2$.
  Then, for all $j,k\geq 0$, $u_j\models\en(\nu_1)$ and $v_k\models\en(\nu_2)$ and thus
  $u_j|v_k\models\en(\nu_1|\nu_2)$ by Lemma~\ref{lem:compose parallel 2}, whenever $u_j|v_k$ is a
  state of\/ $\mathcal{U}$.
  Therefore $\pi\models \mathbf{FG}\,\en(\nu)$.

  Suppose $\pi\Rrightarrow\pi'\backslash c$ (with $c\in\HC$) is $\nu$-enabled because $\nu=\nu'\backslash c$ and
  $\pi'$ is $\nu'$-enabled. Then $c \neq \ell(\nu')\neq \bar c$.
  By induction, $\pi' \models \mathbf{FG}\,\en(\nu')$.
  Using \Lem{compose restriction}, one finds $\pi\models \mathbf{FG}\,\en(\nu)$.

  Suppose $\pi\Rrightarrow\pi'[f]$ is $\nu$-enabled because $\nu=\nu[f]$ and $\pi'$ is $\nu'$-enabled.
  By induction $\pi' \models \mathbf{FG}\,\en(\nu')$.
  Using \Lem{compose relabelling}, one finds $\pi\models \mathbf{FG}\,\en(\nu)$.

  Suppose $\pi$ is $\nu$-enabled because it has a suffix $\pi'$ that is $\nu$-enabled.
  Then, by induction, $\pi' \models \mathbf{FG}\,\en(\nu)$.
  Hence $\pi \models \mathbf{FG}\,\en(\nu)$.
\qed\end{proof}
The following result is the ``$\Leftarrow$''-direction of \Thm{just path}.

\begin{proposition}\label{prop:just path}
If $\pi$ is a path in\/ $\mathcal{U}$ with $\pi\models\mathbf{F}\mathbf{G}\,\en(\nu) \Rightarrow \mathbf{G}\mathbf{F}\, \nu$
for each abstract transition $\nu$ with $\ell(\nu)\in \B!\cup\{\tau\}$, then
$\widehat\pi$ is just in the sense of \Def{just path}.
\end{proposition}

\begin{proof}
By definition no state $P\in\T_{\rm \ABC}$ satisfies $\nu$.
So, any state of $\mathcal{U}$ that could satisfy $\nu$ as well as $en(\nu)$ needs to be a 
derivation $\zeta$. Assume
$\zeta\models en(\nu)$. Then there is a representative $\chi$ of $\nu$, i.e., $\nu=[\chi]_{\equiv}$, with $\chi\conc\zeta$.
Using Lemma~\ref{lm:conc_and_equiv}, we get $\nu=[\chi]_{\equiv}\not=[\zeta]_{\equiv}$.
In case $\zeta$ would also satisfy $\nu$, we have, by definition, $\nu = [\zeta]_\equiv$, which is a contradiction.
Hence there is no abstract transition $\nu$ and
state $u$ in\/ $\mathcal{U}$ for which the propositions $\en(\nu)$ and $\nu$ both hold.
Consequently, the formula $\mathbf{F}\mathbf{G}\,\en(\nu) \Rightarrow \mathbf{G}\mathbf{F}\, \nu$
is equivalent to $\neg\mathbf{F}\mathbf{G}\,\en(\nu)$.

Let $\pi$ be a path in\/ $\mathcal{U}$ with $\pi\models\neg\mathbf{F}\mathbf{G}\,\en(\nu)$
for each $\nu$ with $\ell(\nu)\in \B!\cup\{\tau\}$.
Then, by \Prop{nu-enabled}, $\pi$ is $\nu$-enabled for no $\nu$ with $\ell(\nu)\in \B!\cup\{\tau\}$.
Hence, by \Prop{enabled just}, $\pi$ is $\HC$-just. Therefore, by \Prop{projection}, $\widehat\pi$ is
$\HC$-just, and hence just.
\qed\end{proof}

\subsection{The ``$\Rightarrow$''-direction}

\subsubsection{On the targets of derivations enabling abstract transitions}

\begin{lemma}\label{lem:parallel target}
If $\zeta\models \en(\nu|\_)$, $\zeta\models \en(\_|\nu)$ or $\zeta\models\en(\nu_1|\nu_2)$ for a
derivation $\zeta$ and abstract transitions $\nu,\nu_1,\nu_2$, then $\target(\zeta)$ has the form $P_1|P_2$.
\end{lemma}

\begin{proof}
Since $\zeta\models \en(\nu)$ means that $\chi\conc\zeta$ for a representative $\chi$ of $\nu$, the
lemma can be rephrased as: ``If $\chi\conc\zeta$ for a representative $\chi$ of an abstract
transition $\nu|\_$, $\_|\nu$ or $\nu_1|\nu_2$, $\target(\zeta)$ has the form $P_1|P_2$.''
We prove this statement by structural induction on $\chi$.
\begin{itemize}
\vspace{-1ex}
\item Let $\chi$ be $\chi'| P$,  $P|\chi'$ or $\chi_1|\chi_2$.
By the definition of $\conc$, $\zeta$ must then have the form $Q|\zeta'$, $\zeta'|Q$ or $\zeta_1|\zeta_2$.
In each of these cases $\target(\zeta)$ has the form $P_1|P_2$.
\item Let $\chi=\chi'{+}P$.
  Then $\zeta$ must have the form $\zeta'{+}P$ with $\chi' \conc \zeta'$ by the definition of $\conc$.
  Since $\chi$ is a representative of an abstract transition $\nu|\_$ or $\nu_1|\nu_2$, so is $\chi'$.
  By induction, $\target(\zeta')$ has the form $P_1|P_2$.
  By rule \myref{Sum-l} $\target(\zeta)=\target(\zeta')$, and thus of the form $P_1|P_2$.
\item The cases $\chi=P{+}\chi'$ and $\chi=A{:}\chi'$ proceed likewise.
\end{itemize}
As $\chi$ represents an abstract transition $\nu|\_$ or $\nu_1|\nu_2$, it cannot have any other form.
\qed\end{proof}

\begin{lemma}\label{lem:restriction target}
If $\zeta\models \en(\nu\backslash c)$ or $\zeta\models\en(\nu[f])$, then $\target(\zeta)$ has the
form $P\backslash c$, and $P[f]$, resp.
\end{lemma}

\begin{proof}
The first statement can be rephrased as: ``If $\chi\conc\zeta$ for a representative $\chi$ of an abstract
transition $\nu\backslash c$, then $\target(\zeta)$ has the form $P\backslash c$.''
We prove this statement by structural induction on $\chi$.
\begin{itemize}
\vspace{-1ex}
\item Let $\chi=\chi'\backslash c$. By the definition of $\conc$ we have $\zeta=\zeta'\backslash c$
  with $\chi' \conc \zeta'$. Hence $\target(\zeta)$ has the form $P\backslash c$.
\item Let $\chi=A{:}\chi'$.
  Then $\zeta$ must have the form $A{:}\zeta'$ with $\chi' \conc \zeta'$ by the definition of $\conc$.
  Since $\chi$ is a representative of an abstract transition $\nu\backslash c$, so is $\chi'$.
  By induction, $\target(\zeta')$ has the form $P\backslash c$.
  By rule \myref{Rec} $\target(\zeta)=\target(\zeta')$, and thus of the form $P\backslash c$.
\item The cases $\chi=\chi'{+}P$ and $\chi=P{+}\chi'$ proceed likewise, using  
\myref{Sum-l} and \myref{Sum-r}.
\item As $\chi$ represents an abstract transition $\nu\backslash c$, it cannot have any other form.
\vspace{-1ex}
\end{itemize}
The proof of the second statement proceeds likewise.
\qed\end{proof}

\subsubsection{Decomposing enabled abstract transitions}

\begin{observation}\label{obs:cons33}
Using \Obs{decompose parallel}, any representative $\chi$ of $\nu|\_$
such that \mbox{$\src(\chi)=P_1|P_2$}
is of the form $\chi'|P_2$ with $\chi'$ a representative of $\nu$, 
or $\chi'|\varsigma$ with $\chi'$ a representative of $\nu$ and $\ell(\chi')=\ell(\chi)\in\B!$.
Moreover $\src(\chi')=P_1$.
\end{observation}

\begin{lemma}\label{lem:parallel derivation}
If $u_1|u_2\models \en(\nu|\_)$ for a state $u_1|u_2$ of\/ $\mathcal{U}$, then $u_1 \models \en(\nu)$.
\end{lemma}

\begin{proof}
We make a case distinction based on whether $u_i$ is a process or a derivation.

Suppose $P_1|P_2\models \en(\nu|\_)$. Then $P_1|P_2\mathbin=\src(\chi)$ for a representative $\chi$ of $\nu|\_$.
By \Obs{cons33}, it has the form $\chi'|v$ with $\chi'$ a representative of $\nu$ and $\src(\chi')=P_1$. 
Therefore $P_1\models\en(\nu)$.

Suppose $P_1|\zeta_2\models \en(\nu|\_)$ with $\src(P_1|\zeta_2)=P_1|P_2$. Then $\chi\conc P_1|\zeta_2$ for a
representative $\chi$ of $\nu|\_$. 
By \Obs{source} $\src(\chi)=\src(P_1|\zeta_2)=P_1|P_2$. 
So, by \Obs{cons33} it is has the form $\chi'|v$, with $\chi'$ a representative of $\nu$ and $\src(\chi')=P_1$. 
Therefore $P_1\models\en(\nu)$.

Suppose $\zeta_1|u_2\models \en(\nu|\_)$ with $\src(\zeta_1|u_2)=P_1|P_2$. Then $\chi\conc \zeta_1|u_2$ for a
representative $\chi$ of $\nu|\_$. 
By \Obs{source} $\src(\chi)=\src(\zeta_1|u_2)=P_1|P_2$. So, by \Obs{cons33} it is has either the form $\chi'|P_2$ or $\chi'|\varsigma$ ($\ell(\chi')\in\B!$),
with $\chi'$ a representative of $\nu$.
So, $\chi'|P_2\conc \zeta_1|u_2$ or  $\chi'|\varsigma\conc \zeta_1|u_2$ and hence, by \Def{concurrency}, $\chi'\conc \zeta_1'$. Since $\chi'$ a representative of $\nu$, $\zeta_1\models\en(\nu)$.
\qed
\end{proof}

\begin{lemma}\label{lem:parallel 2 derivation}
If $u_1|u_2\mathbin{\models} \en(\nu_1|\nu_2)$ for a state $u_1|u_2$ of\/ $\mathcal{U}$, then
$u_i\mathbin{\models} \en(\nu_i)$ ($i\mathord=1,2$).
\end{lemma}

\begin{proof}
We make a case distinction based on whether $u_i$ is a process or a derivation.

Suppose $P_1|P_2\models \en(\nu_1|\nu_2)$. Then $P_1|P_2\mathbin=\src(\chi)$ for a representative $\chi$ of $\nu_1|\nu_2$.
By \Obs{decompose parallel}, it has the form $v_1|v_2$.
Since it is a representative of $\nu_1|\nu_2$, it has the form $\chi_1|\chi_2$, with $\chi_i$ a representative of $\nu_i$ ($i\mathord=1,2$).
Furthermore, $\src(\chi_i)=P_i$. It follows that $P_i\models\en(\nu_i)$ ($i\mathord=1,2$).

Suppose $P_1|\zeta_2\models \en(\nu_1|\nu_2)$ with $\src(P_1|\zeta_2)=P_1|P_2$. Then $\chi\conc P_1|\zeta_2$ for a
representative $\chi$ of $\nu_1|\nu_2$. 
By \Obs{source} $\src(\chi)=\src(P_1|\zeta_2)=P_1|P_2$. So,
by \Obs{decompose parallel}, $\chi$ has the form $v_1|v_2$.
Since it is a representative of $\nu_1|\nu_2$, it has the form $\chi_1|\chi_2$, with $\chi_i$ a representative of $\nu_i$ ($i\mathord=1,2$).
Moreover, $\ell(\chi_1)\mathbin=\ell(\nu_1) \mathbin= \overline{\ell(\nu_2)} \mathbin= \overline{\ell(\chi_2)}\mathbin\in\HC$.
So $\chi_1|\chi_2\conc P_1|\zeta_2$ and by \Def{concurrency} $\src(\chi_1)=P_1$ and $\chi_2\conc \zeta_2$. Since $\chi_i$ ($i\mathord= 1,2$) is a representative of $\nu_i$, $P_1\models\en(\nu_1)$ and $\zeta_2\models\en(\nu_2)$.

The case $\zeta_1|P_2$ follows by symmetry.

Suppose $\zeta_1|\zeta_2\models \en(\nu_1|\nu_2)$ with $\src(\zeta_1|\zeta_2)=P_1|P_2$. 
Then $\chi\conc \zeta_1|\zeta_2$ for a representative $\chi$ of $\nu_1|\nu_2$. 
By \Obs{source} $\src(\chi)=\src(P_1|\zeta_2)=P_1|P_2$. So,
by \Obs{decompose parallel}, $\chi$ has the form $v_1|v_2$.
Since it is a representative of $\nu_1|\nu_2$, it has the form $\chi_1|\chi_2$, with $\chi_i$ a representative of $\nu_i$ ($i\mathord=1,2$).
Moreover, $\ell(\chi_1)\mathbin=\ell(\nu_1) \mathbin= \overline{\ell(\nu_2)} \mathbin= \overline{\ell(\chi_2)}\mathbin\in\HC$.
So $\chi_1|\chi_2\conc \zeta_1|\zeta_2$ and by \Def{concurrency} 
$\chi_i\conc \zeta_i$ ($i\mathord= 1,2$). Since $\chi_i$  is a representative of $\nu_i$, $\zeta_i\models\en(\nu_i)$.
\qed\end{proof}

\begin{lemma}\label{lem:restriction derivation}
If $u\models \en(\nu\backslash c)$ for a state $u$ of\/ $\mathcal{U}$ of the form $u'\backslash c$, then $u' \models \en(\nu)$.
\end{lemma}
\begin{proof}
We make a case distinction based on whether $u$ is a processes $P$ or a derivation $\zeta$.

Suppose $P\models \en(\nu\backslash c)$. Then $P\mathbin=\src(\chi)$ for a representative $\chi$ of $\nu\backslash c$.
Since, by assumption, $P$ is of the form $P'\backslash c$, \Obs{decompose restriction} says that $\chi$ is of
the form $\chi'\backslash c$ with $\src(\chi')=P'$. Since $\chi=\chi'\backslash c$ is a representative of $\nu\backslash c$,
$\chi'$ must be a representative of $\nu$.  Hence $P'\models\en(\nu)$.

Suppose $\zeta\models \en(\nu\backslash c)$. Then $\chi\conc\zeta$ for a
representative $\chi$ of $\nu\backslash c$. Since, by assumption, $\zeta$ is of the form $\zeta'\backslash c$,
$\src(\zeta)$ must be of the form $P'\backslash c$. By \Obs{source} $\src(\chi)=\src(\zeta)$, so by
\Obs{decompose restriction} $\chi$ is of the form $\chi'\backslash c$. Since $\chi=\chi'\backslash c$
is a representative of $\nu\backslash c$, $\chi'$ must be a representative of $\nu$.
As $\chi'\backslash c \conc \zeta'\backslash c$, we have $\chi'\conc\zeta'$.
Thus $\zeta' \models \en(\nu)$.
\qed\end{proof}

\begin{lemma}\label{lem:relabelling derivation}
If $u[f]\models \en(\nu[f])$ for a state $u$ of\/ $\mathcal{U}$, then $u \models \en(\nu)$.
\end{lemma}
\begin{proof}
Exactly as above, using an analogue of \Obs{decompose restriction} for relabelling.\!\!\qed
\end{proof}

\subsubsection{$\nu$-enabled paths in\/ $\mathcal{S}$}

\begin{definition}\rm
A path $\rho$ in\/ $\mathcal{S}$ is \emph{$\nu$-enabled} for an abstract transition $\nu$, if
either it is finite and its last state $Q\in\T_{\rm\ABC}$ satisfies $Q\models\en(\nu)$, or
it is infinite and has a suffix $\rho'$ such that 
$\zeta\models\en(\nu)$ for all derivations $\zeta$ with $\widehat\zeta$ a transition in $\rho'$.
\end{definition}

\begin{lemma}\label{lem:nu-enabled}
If a path $\rho$ in\/ $\mathcal{S}$ is $Y\!$-just and $\nu$-enabled then $\ell(\nu)\in Y$.
\end{lemma}

\begin{proof}
If a finite path $\rho$ in\/ $\mathcal{S}$ is $\nu$-enabled, its last state $Q\in\T_{\rm\ABC}$
satisfies $\en(\nu)$, and thus $Q\ar{\ell(\nu)}$. The first clause of \Def{just path}
($Y\!$-justness) tells that $\ell(\nu)\in Y$.

For infinite paths $\rho$, we apply structural induction on $\nu$.
Let $\rho$ be an infinite path that is $Y\!$-just and $\nu$-enabled, and let $\rho'$ be a suffix of
$\rho$, such that $\zeta\models\en(\nu)$ for each derivation $\zeta$ with $\widehat\zeta$ a transition in $\rho'$.
Moreover, $P\models\en(\nu)$ for each state of $P$ on $\rho'$, using \Prop{target} and the definition of $\,\widehat{\ }$.
Let $\zeta_0$ be a derivation of the first transition in $\rho'$,
and let $\rho''$ be the suffix of $\rho'$ starting from $Q:=\target(\zeta_0)$.

Let $\nu=\shar{\alpha}P$ for $\alpha\in Act$ and $P\in\T_{\rm\ABC}$.
Since no representative $\chi$ of $\nu$ is concurrent with any derivation $\zeta$, it follows that
$\rho'$ contains no transitions, and hence consists of a single state only.
This contradicts the presumed infinity of $\rho$.

Let $\nu\mathbin=\nu_1|\_$. Since $\zeta_0\models\en(\nu_1|\_)$, by \Lem{parallel target} $Q$ has the form $P_1|P_2$.
By \Def{just path} $\rho''$ can be decomposed into an $X$-just path $\rho_1$ of $\widehat P_1=P_1$ and a $Z$-just
path $\rho_2$ of $\widehat P_2=P_2$ such that $Y\supseteq X\mathord\cup Z$ and $X\mathord\cap \bar{Z}\mathbin=\emptyset$.
By \Obs{decompose parallel} and the definition of $\,\widehat\ $, all states $u$ of $\rho''$ as well as all derivations $u$ of transitions in $\rho''$ have the form $u_1|u_2$.
Since each such $u_1|u_2$ satisfies $\en(\nu_1|\_)$, by \Lem{parallel derivation} $u_1\models\en(\nu_1)$.
It follows that $u_1\models\en(\nu_1)$ for each state $u_1$ in $\rho_1$ and for each derivation $u_1$ of a transition in $\rho_1$.
Hence $\rho_1$ is $\nu_1$-enabled. 
By induction $\ell(\nu_1)\in X$. So $\ell(\nu)=\ell(\nu_1)\in X\subseteq Y$.

The case $\nu\mathbin=\_|\nu_2$ follows by symmetry.

Let $\nu\mathbin=\nu_1|\nu_2$.  Then $\ell(\nu_1)=\overline{\ell(\nu_2)}\in\HC$.
Since $\zeta_0\models\en(\nu_1|\nu_2)$, by \Lem{parallel target} $Q$ has the form $P_1|P_2$.
By \Def{just path} $\rho''$ can be decomposed into an $X$-just path $\rho_1$ of $\widehat P_1=P_1$ and a $Z$-just
path $\rho_2$ of $\widehat P_2=P_2$ such that $Y\supseteq X\mathord\cup Z$ and $X\mathord\cap \bar{Z}\mathbin=\emptyset$.
By \Obs{decompose parallel} and the definition of $\,\widehat\ $, all states $u$ of $\rho''$ as well as all derivations $u$ of transitions in $\rho''$ have the form $u_1|u_2$.
Since each such $u_1|u_2$ satisfies $\en(\nu_1|\nu_2)$, by \Lem{parallel 2 derivation} $u_i\models\en(\nu_i)$ ($i\mathord=1,2$).
It follows that $u_i\models\en(\nu_i)$ for each state $u_i$ in $\rho_i$ and for each derivation $u_i$ of a transition in $\rho_i$.
Hence $\rho_i$ is $\nu_i$-enabled. By induction $\ell(\nu_1)\in X$ and
$\ell(\nu_1)=\overline{\ell(\nu_2)}\in \bar{Z}$, in
contradiction with $X\mathord\cap \bar{Z}\mathbin=\emptyset$. Therefore, this case cannot occur.

Let $\nu\mathbin=\nu'\backslash c$ (with $c\mathbin\in\HC$). Then $c \neq\ell(\nu')\neq \bar c$.
Since $\zeta_0\models\en(\nu'\backslash c)$, by \Lem{restriction target} $Q$ has the form $P\backslash c$.
By \Def{just path} $\rho''$ can be decomposed into a $Y\mathord\cup\{c\}$-just path $\rho'''$ of $\widehat P=P$.
By \Obs{decompose restriction} all derivations $\zeta$ of
transitions in $\rho''$ have the form $\zeta'\backslash c$.
Since each such $\zeta'\backslash c$ satisfies $\en(\nu'\backslash c)$, by \Lem{restriction derivation} $\zeta'\models\en(\nu')$.
It follows that $\zeta'\models\en(\nu')$ for each derivation $\zeta'$ of a transition in $\rho'''$.
Hence $\rho'''$ is $\nu'$-enabled. By induction $\ell(\nu')\in Y\cup\{c\}$. Since $\ell(\nu')\neq c$
we obtain $\ell(\nu)=\ell(\nu')\in Y$.

Let $\nu\mathbin=\nu'[f]$.
Since $\zeta_0\models\en(\nu'[f])$, by \Lem{restriction target} $Q$ has the form $P[f]$.
By \Def{just path} $\rho''$ can be decomposed into a $f^{\-1}(Y)$-just path $\rho'''$ of $\widehat P=P$.
By the relabelling variant of \Obs{decompose restriction} all derivations $\zeta$ of
transitions in $\rho''$ have the form $\zeta'[f]$.
Since each such $\zeta'[f]$ satisfies $\en(\nu'[f])$, by \Lem{relabelling derivation} $\zeta'\models\en(\nu')$.
It follows that $\zeta'\models\en(\nu')$ for each derivation $\zeta'$ of a transition in~$\rho'''$.
Hence $\rho'''$ is $\nu'$-enabled. By induction $\ell(\nu')\in f^{\-1}(Y)$. So
$\ell(\nu)=f(\ell(\nu'))\in Y$.
\qed
\end{proof}

\noindent
The following result directly implies the ``$\Rightarrow$''-direction of \Thm{just path}.

\begin{proposition}
Let $\rho$ be a just path in\/ $\mathcal{S}$. Then $\rho=\widehat\pi$ for a path $\pi$ in\/ $\mathcal{U}$ that satisfies
$\pi\not\models \mathbf{FG}\,\en(\nu)$ for each abstract transition $\nu$ with $\ell(\nu)\in\B!\cup\{\tau\}$.
\end{proposition}

\begin{proof}
First, suppose that $\rho$ is a finite just path in\/ $\mathcal{S}$.
By \Def{just path} it ends in a state $Q\in\T_{\rm\ABC}$ that admits actions from $\HC\cup\B?$ only.
Pick any path $\pi$ in\/ $\mathcal{U}$ with $\widehat\pi =\rho$.
Then $\pi$ ends in $Q$ as well.
Hence $\pi\models \mathbf{FG}\,\en(\nu)$ for no abstract transition $\nu$ with $\ell(\nu)\in\B!\cup\{\tau\}$.

Next, consider the case that $\rho$ is infinite.
There are countably many abstract transitions.
Let $(\nu_i)_{i=0}^\infty$ be an enumeration of the abstract transitions $\nu$ with $\ell(\nu)\in\B!\cup\{\tau\}$,
such that each such $\nu$ occurs infinitely often in this sequence.

With induction on $i\in \IN$, we construct finite paths
$\pi_i$ in\/ $\mathcal{U}$ such that $\pi_i$ will be a strict prefix of $\pi_j$ when $i<j$, and
$\widehat\pi_i$ is a prefix of $\rho$ for each $i\in\IN$.

Let $\pi_0$ be an arbitrary finite path in\/ $\mathcal{U}$ with $\widehat\pi_0$ a prefix of $\rho$.
Given $\pi_i$, let $\zeta_i$ be an arbitrary derivation 
such that $\zeta_i\not\models \en(\nu_i)$ and $\widehat \zeta_i$ occurs in $\rho$ past
the prefix $\widehat\pi_i$.
Such a $\zeta_i$ must exists, as otherwise $\rho$ would be $\nu_i$-enabled, which contradicts \Lem{nu-enabled}.
(Remember that by assumption $\rho$ is just.)

We obtain $\pi_{i+1}$ by extending $\pi_i$ in a way such that $\widehat\pi_{i+1}$ is a prefix of 
$\rho$ up to and including $\widehat\zeta_i$ and its target state; the last derivation of $\pi_{i+1}$ is set  to $\zeta_i$.
All derivations different from $\zeta_i$ that are not part of $\pi_i$ can be chosen
arbitrarily, under the restriction that $\widehat\pi_{i+1}$ is a prefix of $\rho$.

Now $\pi:=\lim_{i\rightarrow\infty}\pi_i$ exists and satisfies $\rho=\widehat\pi$. By construction,
$\pi\models \neg\mathbf{FG}\,\en(\nu)$ for any abstract transition $\nu$ with $\ell(\nu)\in\B!\cup\{\tau\}$.
\qed\end{proof}

\end{appendix}
\end{document}